\documentclass[a4paper,11pt]{article} 
\usepackage{color}
\usepackage{amsmath}
\usepackage{amsthm}
\usepackage{amssymb}
\usepackage{amsmath}
\usepackage{mathcomp}
\usepackage{dsfont}
\usepackage[toc,page]{appendix}
\usepackage{graphicx} 
\usepackage{subfigure}
\usepackage{enumerate}
\usepackage{amsfonts}
\usepackage{amsthm}
\usepackage{amsmath}
\usepackage{amscd}
\usepackage{t1enc}
\usepackage{enumitem}
\usepackage[mathscr]{eucal}
\usepackage{indentfirst}
\usepackage{graphicx}
\usepackage{graphics}
\usepackage{pict2e}
\usepackage{epic}
\usepackage{esint}
\usepackage{epstopdf} 
\usepackage{verbatim}
\usepackage{cancel}
\usepackage{amssymb}
\usepackage{xcolor}
\usepackage{dsfont}
\usepackage{hyperref}
\usepackage{physics}
\allowdisplaybreaks

\definecolor{Greenish}{RGB}{34, 139 , 34}
\definecolor{Blueish}{RGB}{39,64, 139}

\usepackage{hyperref}
\hypersetup{
    colorlinks=true,
    linkcolor=Greenish,
    citecolor=magenta,
    }

\usepackage[margin=30mm]{geometry}
    
\newtheorem{theorem}{Theorem}[section]
\newtheorem{lemma}[theorem]{Lemma}
\newtheorem{proposition}[theorem]{Proposition}

\newtheorem{corollary}[theorem]{Corollary}

\theoremstyle{definition}

\numberwithin{equation}{section}

\def\b1{\mathds{1}}
\def\1{\mathds{1}}

\def\Re{\mathrm{Re}}
\def\Im{\mathrm{Im}}

\def\a0{\mathrm{a}_0}

\def\t{\mathrm{t}}

\begin{document}
\title{Dispersive estimates and long-time validity for Bogoliubov dynamics of interacting Bose gases}

\author{Phan Th\`{a}nh Nam \thanks{Department of Mathematics, LMU Munich, Theresienstr. 39, 80333 Munich,
Germany}  \and Simone Rademacher \footnotemark[1] \and Avy Soffer \thanks{Rutgers University, Mathematics department, 110 Frelinghuysen Road, Piscataway, NJ 08854, US}}

\date{\today}

\maketitle

\begin{abstract} 
We consider the Bogoliubov approximation for the many-body quantum dynamics of weakly interacting Bose gases and establish a uniform-in-time validity of the Bogoliubov theory. The proof relies on a detailed analysis of the dispersive behavior of the symplectic Bogoliubov dynamics, which allows for a rigorous derivation of the Bogoliubov theory as an effective description of quantum fluctuations around the Bose--Einstein condensate on all time scales.
\end{abstract}

\section{Introduction}

As theoretically predicted by Bose \cite{Bose} and Einstein \cite{Einstein} and later experimentally observed \cite{Wieman,Ketterle}, trapped Bose gases at low temperatures undergo a peculiar phase transition and form a Bose-Einstein condensate. In this state, most of the interacting bosons occupy the same quantum state, known as the condensate.

Typical experimental setups study the time evolution of a condensate formed from an initially trapped, cold Bose gas after the trap is released. While the dynamics of the condensate is well understood, the dynamics of the particles outside the condensate, known as quantum depletion and theoretically predicted by Bogoliubov theory \cite{Bogo}, have also become accessible to experiments \cite{Ketterle2} in recent years.

This paper is dedicated to the mathematical analysis of the dynamics in Bose–Einstein condensates based on the first principles of quantum physics, specifically many-body quantum systems. In the large particle limit, we mathematically prove that Bogoliubov dynamics effectively describes the quantum depletion for all times, under suitable assumptions on the initial state. Our results complement recent discussions in the physics literature concerning the survival of quantum depletion after the trap is released (see, for example, \cite{Ketterle3,QPS} and \cite{Ross}) by providing a mathematical framework for the long-time validity of Bogoliubov theory. 

\subsection{Mathematical setting}

From the first principles of quantum physics, the dynamics of a Bose gas of $N$ particles in $\mathbb{R}^3$ is typically described by the Schr\"odinger equation 
\begin{align}
\label{eq:Schroe}
i \partial_t \psi_{N,t} = H_N \psi_{N,t} 
\end{align}
where $H_N$ is the Hamiltonian acting on $L^2_s(\mathbb{R}^{3N})$,  the symmetric subspace of $L^2(\mathbb{R}^{3N})$ consisting of wave functions that are invariant under all permutations of the particle coordinates in $\mathbb{R}^3$. In the present paper, we are interested in weakly interacting bosons with mean-field type pairwise interactions, where
\begin{align}
H_N = \sum_{i=1}^N (-\Delta_i) + \frac{1}{N} \sum_{1\leq i<j\leq N} v (x_i -x_j)  \;  
\end{align}
with a regular even potential $v:\mathbb{R}^3 \to \mathbb{R}$. 

At very low temperatures, Bose gasses undergo a remarkable phase transition from an ordinary gas to a novel state of matter, in which the majority of the interacting atoms occupy a single quantum state. Mathematically, since the interaction is sufficiently weak, it is natural to focus on the ansatz where the particles are essentially independent, leading to the consideration of factorized states of the form
\begin{align}
\psi_{N,t} \approx \varphi_t^{\otimes N}, \quad \text{with} \quad \varphi_t \in L^2(\mathbb{R}^3),
\end{align}
which are often referred to as pure condensates. Although the dynamics does not preserve the exact factorization property, it is well known that the majority of the $N$ interacting particles remain in the same one-particle state for all positive times $t>0$. In fact, due to the mean-field scaling of the interaction potential, the evolution of the condensate $\varphi_t \in L^2(\mathbb{R}^3)$  is governed by a nonlinear Schr\"odinger equation of the form
\begin{align}
\label{def:Hartree}
i \partial_t \varphi_t = h_H(t) \varphi_t, \quad h_H(t) = \big( - \Delta + v * |\varphi_t|^2 \big),
\end{align}
with initial data $\varphi_0 \in H^2(\mathbb{R}^3)$ normalized by  $\| \varphi_0 \|_{L^2} = 1$ (in this paper we write  $\| f \|_{L^p ( \mathbb{R}^3)} = \| f \|_{L^p}$ for any $p \geq 1.$). The derivation of the Hartree equation \eqref{def:Hartree} from the many-body dynamics \eqref{eq:Schroe} for the condensate's effective description in trace norm topology in the large particle limit, $N \rightarrow \infty$, is widely studied in the literature under suitable assumptions on the potential \cite{GV3,GV4,GV2,H,S}, singular potentials \cite{BGM,EY,RS,KP,P}, for factorized initial data \cite{RS} and derived with optimal rate of convergence $\mathcal{O} (N^{-1})$ in \cite{ES,CLS,CL,L}, and  {remarkably even for all times in recent works \cite{CLS,DL}. }

It is well known that, under suitable assumptions on $\varphi_0$ and potential $v$, the condensate $\varphi_t$ disperses in time, i.e. $\| \varphi_t\|_{L^\infty} \to 0$ as $t \to \infty$; see, e.g., \cite{CW,CO,DZ,DHR,GO,GV1,GV2,GM1,HT,O}. In particular, we will recall in Proposition \ref{prop:dispersive} below  a dispersive estimate from \cite{GM1} which is the starting point of our analysis.

In the present paper, we go one step further, focusing on the quantum fluctuations around the Hartree dynamics. For this purpose, following \cite{LNSS,LNS} we factor out the contribution of the condensate by a unitary transformation mapping from the original $N$-body Hilbert space to the Fock space of excited particles. To be precise, as observed in \cite{LNSS}, any $N$-particle function $\psi_N \in L^2_s ( \mathbb{R}^{3N})$ can be uniquely decomposed as 
\begin{align}
\psi_N = \sum_{j=0}^N \varphi_t^{\otimes (N-j)} \otimes_s \xi_t^{(j)}
\end{align}
and mapped on its so-called excitation vector $\xi_{N,t} := \lbrace \xi_t^{(1)}, \dots, \xi_t^{(N)} \rbrace$ by the unitary 
\begin{align}
\label{def:U}
\mathcal{U}_{N,t} : L_s^2( \mathbb{R}^{3N} ) \rightarrow \bigoplus_{n=0}^N L^2_{\perp \varphi_t}( \mathbb{R}^{3})^{\otimes n}, \quad \mathcal{U}_{N,t}\psi_{N,t} \mapsto \lbrace \xi_t^{(0)}, \dots, \xi_t^{(N)} \rbrace 
\end{align}
where $L^2_{\perp \varphi_t}( \mathbb{R}^{3N})$ denotes the orthogonal complement of the solution to the Hartree equation $\varphi_t$ (i.e. the condensate) in $L^2( \mathbb{R}^{3N})$. The fluctuations around the Hartree evolution are described through the fluctuation dynamics 
\begin{align}\label{def:flucdyn}
\mathcal{W}_N (t;s) = \mathcal{U}_{N,t} e^{-i (t-s)H_N} \mathcal{U}_{N,s}^* \; 
\end{align}
that is an element of the truncated Fock space 
\begin{align}
\mathcal{F}_{\perp \varphi_t}^{\leq N} := \bigoplus_{n=0}^N L^2_{\perp \varphi_t}( \mathbb{R}^{3})^{\otimes n} \; . 
\end{align}

For every finite time $t \ge s$, it was proved in \cite{LNS} that, in the limit $N \to \infty$, the fluctuation dynamics $\mathcal{W}_N(t;s)$ converges to the Bogoliubov dynamics, a simplified model that serves as an effective framework for describing second-order corrections to the Bose-Einstein condensate. Our goal is to show that the Bogoliubov dynamics, in a suitable symplectic representation, exhibits a strong dispersive behavior. As a consequence, we establish that its validity as an approximation to the full many-body dynamics in \eqref{def:flucdyn} holds {\em uniformly in time}.

\subsection{Symplectic Bogoliubov dynamics} 

In the physics literature, quantum fluctuations around a Bose-Einstein condensate are described by a \textit{ symplectic Bogoliubov dynamics}, that is a two-parameter family of operators $\Theta_{(t;s)}$ on $ L^2( \mathbb{R}^3) \oplus L^2( \mathbb{R}^3)$ given in terms of two bounded linear operators $\gamma_{(t;s)}, \sigma_{(t;s)}$ on $L^2( \mathbb{R}^3)$ through 
\begin{align}
\label{def:gamma,sigma}
\Theta_{(t;s)} = \begin{pmatrix}
\gamma_{(t;s)} & \overline{\sigma_{(t;s)}} \\
\sigma_{(t;s)} & \overline{\gamma_{(t;s)}} 
\end{pmatrix} , \quad \text{with} \quad \gamma_{(t;t)} (x;y)  = \delta (x-y) , \quad \text{and} \quad \sigma_{(t;t)} = 0 \; . 
\end{align}
The symplectic Bogoliubov dynamics satisfies the differential equation 
\begin{align}
i \partial_s \Theta_{(t;s)} = \mathcal{T}_s \Theta_{(t;s)} , \qquad \text{with} \quad \mathcal{T}_s = \begin{pmatrix}
 H_s & - K_s \\
\overline{K}_s & -  H_s  
\end{pmatrix} \; . \label{def:bogo}
\end{align}
where the operators 
\begin{align}
\label{def:Hs,Ks}
H_s = h_H (s) + \widetilde{K}_{1,s}, \quad \text{and} \quad K_s =  \widetilde{K}_{2,s}
\end{align}
are given in terms of the Hartree Hamiltonian $h_H (s)$ from \eqref{def:Hartree} and the operators $\widetilde{K}_{1,s}, \widetilde{K}_{2,s}$ having kernels $\widetilde{K}_{1,s} (x;y) = (q_s K_1 q_s) (x;y)$ resp. $\widetilde{K}_{2,s}(x;y) = (\overline{q}_s K_{2,s} q_s) (x;y)$ that are through the projection $q_s = 1- \vert \varphi_s \rangle \langle \varphi_s \vert$ orthogonal to the condensate's $\varphi_s$ dynamics \eqref{def:Hartree} and given by 
\begin{align}
\label{def:K} 
K_{1,s} (x;y) = v(x-y) \varphi_s (x) \overline{\varphi}_s (y) , \quad K_{2,s} (x;y) = v(x-y) \varphi_s (x) \varphi_s (y) \; . 
\end{align}

As shown in \cite{NM1,R}, the Bogoliubov dynamics can be characterized by its explicit action on creation and annihilation operators. More precisely,  the dynamics $\mathcal{W}_{2,t}$ acts on the Fock space operator  
\begin{align}
A(f,g) = a^*(f) + a(\overline{g}), \quad f,g \in L^2(\mathbb{R}^3),
\end{align}
according to  
\begin{align}
\label{eq:action-bogo}
\mathcal{W}_2^*(t;s) \, A(f,g) \, \mathcal{W}_2(t;s) = A\big( \Theta(t;s)(f,g) \big),
\end{align}
where $\Theta(t;s)$ denotes the symplectic Bogoliubov dynamics defined in \eqref{def:bogo}. This identity provides a direct way to compute the expectations of certain operators with respect to quasi-free states in terms of the operators $\gamma_{(t;s)}$ and $\sigma_{(t;s)}$ introduced in \eqref{def:bogo}, thus establishing a clear link between the symplectic structure of the one-body operators in \eqref{def:bogo} and the Bogoliubov dynamics in Fock space.

Our first result shows that the symplectic Bogoliubov dynamics satisfies dispersive estimates. To formulate its dispersive behavior, we define for an operator $\sigma_{(t;s)}$ on $L^2( \mathbb{R}^3)$ the norm 
\begin{align}
\label{def:norm}
\| \sigma_{(t;s)} \|_{L^\infty \times L^2}^2 := \sup_{x \in \mathbb{R}^3} \| \sigma_{(t;s)} (x, \cdot ) \|_{L^2}^2 \; . 
\end{align}

\begin{theorem}[Dispersive estimate for  Bogoliubov dynamics]
\label{thm:bogo-disp}
Let $v\in C_0^2( \mathbb{R}^3)$ be a nonnegative, radially symmetric and decreasing function. Let $\varphi_t$ be the ($L^2$-normalized) solution to the Hartree equation \eqref{def:Hartree} with initial data $\varphi_0 \in W^{1,\ell} ( \mathbb{R}^3)$ for sufficiently large $\ell>0$. Then   for all $0 \leq s \leq t$, we have 
\begin{align}
\| \sigma_{(t;s)} \|_{L^\infty \times L^2} \leq \frac{C}{(\vert t-s\vert +1 )^{3/2}} 
 \end{align}
 with a constant $C>0$ independent of $t,s$. 
  
\end{theorem}

On the one hand, the starting point of our proof of Theorem \ref{thm:bogo-disp} is the dispersive properties of the Hartree equation established in \cite{GM1}; in particular, the assumptions on the interaction potential $v$ and the initial datum $\varphi$ are taken from \cite[Proposition 3.1]{GM1} (see also Proposition \ref{prop:dispersive} below). 

On the other hand, the dispersive properties of the symplectic Bogoliubov dynamics, expressed in terms of the norm $\|\sigma_{(t;s)}\|_{L^\infty \times L^2}$, seem to appear here for the first time in Theorem \ref{thm:bogo-disp}. We refer to the earlier works \cite{GM1, NM1} for analysis of the long-time behaviour of $\gamma{(t;s)}$ and $\sigma_{(t;s)}$, although using norms different from those in \eqref{def:norm} (in particular, in \cite{GM1, NM1}, the norm $\|\sigma_{(t;s)}\|_{L^2 \times L^2}$ is shown to grow only logarithmically in time -- this will be improved in our Lemma \ref{lemma:gamma,sigma} below).

\subsection{Many-body Bogoliubov dynamics}

Next we apply the previous dispersive result to show that  the fluctuation dynamics is in the large particle limit $N \rightarrow \infty$ well approximated by a \textit{many-body Bogoliubov dynamics} $\mathcal{W}_2(t;0)$: 
A \textit{many-body Bogoliubov dynamics} is unitary map on the full bosonic Fock space $\mathcal{F} := \bigoplus_{n \geq 0} L_s^2( \mathbb{R}^{3N})^{\otimes_s^n}$ with quadratic generator, i.e. $\mathcal{W}_2(t;0)$ that satisfies 
\begin{align}
\label{def:bogo-asymp}
i \partial_t \mathcal{W}_2 (t;0) = \mathbb{H}_t \mathcal{W}_2 (t;0) 
\end{align}
with $\mathcal{W}_2 (0;0) =1$ and 
\begin{align}
\label{def:H}
\mathbb{H}_t  = d\Gamma( H_t) + \frac{1}{2} \int dxdy \; \big[ K_{t} (x;y) a_x^*a_y^* + {\rm h.c.} \big] \, 
\end{align}
where $a_x^*,a_x$ denote the point-wise creation and annihilation operators on the bosonic Fock space. In our setting, we have $H_t$ and $K_t$ are given by \eqref{def:Hs,Ks}.

In our second result, we make the physicists' heuristics mathematically rigorous and prove that the Bogoliubov dynamics serves as an effective description for quantum fluctuations in Bose-Einstein condensates for all times.

\begin{theorem}[Uniform-in-time validity of Bogoliubov approximation] 
\label{thm:norm}
Let $v \in C_0^2 ( \mathbb{R}^3)$ be non-negative, spherically symmetric, decreasing and $\varphi_t$, $t \in \mathbb{R}$, denote the solution to the Hartree equation \eqref{def:Hartree} with initial data $\varphi_0  \in W^{\ell,1}( \mathbb{R}^3)$ for sufficiently large $\ell >0$. Moreover, assume that $\psi_{N,t}$ denotes the solution to the Schr\"odinger equation with initial data $\psi_{N,0} = \varphi_0^{\otimes N}$. 
Let $\mathcal{U}_{N,t}$ be given by \eqref{def:U} and $\mathcal{W}_2(t;0)$ denote the Bogoliubov dynamics defined in \eqref{def:bogo-asymp}. Then, for some $C>0$, we have 
\begin{align}
\| \mathcal{U}_{N,t} \psi_{N,t} - \mathcal{W}_2 (t;0) \Omega \|^2 \leq   \frac{C}{N} \quad \text{for all} \quad t \in \mathbb{R} \; . 
\end{align}
\end{theorem}

The validity of the Bogoliubov approximation for mean-field bosons in the large particle limit on time scales $t= o(\ln (N))$ is well known and widely studied in the literature under suitable assumptions on the initial data and the potential, see for example \cite{BPPS,LNS,NM1,NM2,MPP,GM1,GM2,GMM1,GMM2}. Typically, the restriction $t \ll \ln (N)$ comes from the use of a Gr\"onwall-type estimate, leading to error bounds of the form $e^{Ct} N^{-1}$. In general, these bounds can be improved by using dispersive estimates for the Hartree dynamics, and in fact  the Bogoliubov approximation  has been justified for times scales $t \ll \sqrt{N}$; see, e.g. \cite{GM1,GM2,NM1}.

Thus, our Theorem \ref{thm:norm} significantly strengthens the existing results and validates Bogoliubov's approximation for all times $t \in \mathbb{R}$. At this point, let us remark that even if one could show that the difference between the $N$-body Hamiltonian $H_N$ and the Bogoliubov generator $\mathbb{H}_t$ in \eqref{def:H} is bounded by $\mathcal{N}^2/N$ due to the mean-field scaling, then a direct application of such a simple bound—as done in some previous works, see, e.g., \cite{NM1}—would typically lead to an error of order $t/N$ in the norm approximation, which is meaningful only for times $t < N$. Therefore, the norm approximation in Theorem \ref{thm:norm} goes substantially beyond a simple combination of the dispersive estimates for the Hartree dynamics and the mean-field scaling. In particular, the proof of Theorem \ref{thm:norm} relies crucially on the dispersive estimate for the Bogoliubov dynamics established in Theorem \ref{thm:bogo-disp}.

Moreover, if we wait for a sufficiently long time, the dispersive behavior of the many-body Bogoliubov dynamics becomes dominant, implying that the Bogoliubov dynamics itself can be approximated by the free evolution $e^{-it d\Gamma(\Delta)}$. As an immediate consequence of Theorem \ref{thm:norm}, we thus obtain the following norm approximation of the fluctuation dynamics by the free evolution of a quasi-free state, which remains valid on time scales of order $o(N^2)$.

\begin{corollary}
\label{cor:largetimes}
Under the same assumptions of Theorem \ref{thm:norm}, we have for some $C_1, C_2 >0$
\begin{align}
\| \mathcal{U}_{N,t} \psi_{N,t} - e^{-i(t-t_0)d\Gamma ( \Delta )} \mathcal{W}_2 (t_0,0)\Omega \|^2 \leq \frac{C_1}{N} \quad \text{for all} \quad t \geq t_0 \geq C_2 N^2 \; . 
\end{align}
\end{corollary}

To our knowledge, a uniform-in-time derivation of the condensate from the many-body Schr\"odinger equation in trace norm topology is available in the remarkable work \cite{DL,CLS}. However, its extension to the Bogoliubov theory, that is, going from a first-order to a second-order result, has remained unknown for a while. We hope that our findings will stimulate further investigations into the long-time behavior of quantum dynamics, in particular concerning its connection to the kinetic equation as proposed in \cite{ESY}.

\bigskip
\noindent{\bf Structure of the paper.} The remainder of the paper is organized as follows. In Section~\ref{sec:symplbogo}, we recall and establish key properties of the Hartree and symplectic Bogoliubov dynamics, which serve as the foundation for the proof of Theorem~\ref{thm:norm}. Section~\ref{sec:mbbogo} is devoted to the proof of Theorem~\ref{thm:bogo-disp}, preceded by a discussion of the properties of the fluctuation dynamics and the asymptotic many-body Bogoliubov dynamics.

\section{Symplectic Bogoliubov dynamics}
\label{sec:symplbogo}

The goal of this section is the proof of Theorem \ref{thm:bogo-disp}. To this end, we first recall resp. prove properties of the Hartree dynamics resp. the symplectic Bogoliubov dynamics in Section \ref{sec:Hartree} resp. Section \ref{sec:symplbogo}. Then we prove Theorem \ref{thm:bogo-disp} and Corollary \ref{cor:largetimes} in Section \ref{sec:proofthmdisp} resp. Section \ref{sec:cor}. 

\subsection{Properties of Hartree dynamics}
\label{sec:Hartree}

In section we collect properties of the solution of the Hartree equation \eqref{def:Hartree}, in particular on its dispersive behaviour and consequences for the kernels $\widetilde{K}_{1,s}, \widetilde{K}_{2,s}$ defined in \eqref{def:K}.  The following Proposition was proven in \cite{GM1}. 

\begin{proposition}[Proposition 3.3 in \cite{GM1}]
\label{prop:dispersive}
Let $v \in C_0^1 ( \mathbb{R}^3)$, $t \in \mathbb{R}$ non-negative, spherically symmetric, decreasing and $\varphi_t$ denote the solution to the Hartree equation with initial data $\varphi_0 \in W^{1, \ell }( \mathbb{R}^3)$ for sufficiently large $\ell>0$. Then, there exists a constant $C>0$ (that depends on $\| \varphi_0 \|_{W^{1,\ell}( \mathbb{R}^3)}$ only) such that
\begin{align}
\| \varphi_t \|_{L^\infty} \leq\frac{C}{(1+ \vert t \vert )^{3/2}} \; . 
\end{align}
Moreover, for $s \in \mathbb{N}$ there exists $C_s >0$ (that depends on $\| \varphi_0 \|_{H^s}$ only) such that for all $t \in \mathbb{R}$
\begin{align}
\| \varphi_t \|_{H^k( \mathbb{R}^3)} \leq C_s  \; . 
\end{align} 
\end{proposition}

Next we bound the operator and Hilbert-Schmidt norm of the operator $\widetilde{K}_{j,t}$ in terms of the $L^\infty( \mathbb{R}^3)$-norm of the solution of the Hartree equation and then use Proposition \ref{prop:dispersive} to deduce properties of $\widetilde{K}_{j,t}$. 

\begin{lemma}
\label{lemma:K}
Let $v \in L^1( \mathbb{R}^3) \cap L^2( \mathbb{R}^3)$, $t \in \mathbb{R}$, $\varphi_t$ denote the solution to the Hartree equation with initial data $\varphi_0 \in H^1( \mathbb{R}^3)$ with $\| \varphi_0 \|_{L^2} =1$ and $\widetilde{K}_{1,t}, \widetilde{K}_{2,t}$ be given by \eqref{def:K}. Then, for $j=1,2$ 
\begin{align}
\label{eq:boundK-allg}
\| \widetilde{K}_{j,t} \|_{\rm op} \leq  \| v \|_{L^1} \| \varphi_t \|_{L^\infty}^2 \; \quad \text{and} \quad \| \widetilde{K}_{j,t} \|_{\rm HS} \leq \| v \|_{L^2} \| \varphi_s \|_{L^\infty} \| \varphi_s \|_{L^2}\; . 
\end{align}
and 
\begin{align}
\| \nabla \widetilde{K}_{j,s} \|_{\rm HS}
\leq&  \| \varphi_s \|_{L^\infty} \bigg( \| \nabla \varphi_s \|_{L^2} \| v \|_{L^2} + \| \nabla v \|_{L^2} \| \varphi_s \|_{L^2} \bigg)   \notag \\
\| \Delta \widetilde{K}_{j,s} \|_{\rm HS} \leq&  \| \varphi_s \|_{L^\infty} \bigg( \| \Delta v \|_{L^2} \| \varphi_s \|_{L^2} + \| \nabla v \|_{L^2} \| \nabla \varphi_s \|_{L^2}  + \| v \|_{L^2} \| \Delta \varphi_s \|_{L^2} \bigg) 
\end{align}
Moreover,  let $v \in C_0^2 ( \mathbb{R}^3)$, $t \in \mathbb{R}$ and $\varphi_t$ denote the solution to the Hartree equation with initial data $\varphi_0 \in W^{1, \ell }( \mathbb{R}^3)$ for sufficiently large $\ell>0$. Then, there exists a constant $C>0$ (that depends on $\| \varphi_0 \|_{W^{1,\ell}( \mathbb{R}^3)}$ only) such that
\begin{align}
\label{eq:boundK-disp}
\| \widetilde{K}_{j,t} \|_{\rm op}  \leq \frac{C \| v \|_{L^1}}{( 1 +\vert t \vert )^{3}}, \quad \| \widetilde{K}_{j,t} \|_{\rm HS}  \leq \frac{C \| v \|_{L^2}}{( 1 +\vert t \vert )^{3/2}} \; 
\end{align}
and 
\begin{align}
\|\nabla  \widetilde{K}_{j,t} \|_{\rm HS}  \leq \frac{ C \| v \|_{H^1}}{ (1 + \vert t \vert )^{3/2}}, \quad  \|\Delta  \widetilde{K}_{j,t} \|_{\rm HS}  \leq \frac{ C \| v \|_{H^2}}{ (1 + \vert t \vert )^{3/2}} \; . 
\end{align}
\end{lemma}

\begin{proof}
We first prove \eqref{eq:boundK-allg}. Then \eqref{eq:boundK-disp} follows immediately from Proposition \ref{prop:dispersive}. 

We start with the bound for the operator norm and first observe that since 
\begin{align}
\| \widetilde{K}_{1,s} \|_{\rm op} \leq \| K_{1,s} \|_{\rm op} 
\end{align} 
it suffices to control the operator norm of $K_{1,s}$. For this, let $f,g \in L^2( \mathbb{R}^3)$. Then we have 
\begin{align}
\vert \langle f, K_{1,s} g \rangle \vert =& \bigg\vert \int dxdy  \overline{f}(x) \varphi_s (x)  v (y-x)  \varphi_s (y)  \vert g ( y ) \bigg\vert \notag \\
\leq& \bigg( \int dxdy \; v(x-y) \vert \varphi_s (y) \vert^2 \vert f(x) \vert^2 \bigg)^{1/2}\bigg( \int dxdy \; v(x-y) \vert \varphi_s (x) \vert^2 \vert g(y) \vert^2 \bigg)^{1/2} \notag \\
\leq& \| v \|_{L^1} \| \varphi_s \|_{L^\infty}^2 \| f \|_{L^2} \| g \|_{L^2} 
\end{align}
and thus 
\begin{align}
\| K_{1,s} \|_{\rm op} \leq \| v \|_{L^1} \| \varphi_s \|_{L^\infty}^2  \; . 
\end{align}
The operator norm of $\widetilde{K}_{2,s}$ can be estimated in a similar way. 

Similarly, for the Hilbert-Schmidt norm we have for $j=1,2$
\begin{align}
\| \widetilde{K}_{j,s} \|_{\rm HS} \leq \| K_{j,s} \|_{\rm HS}
\end{align}
and thus, it remains to control the Hilbert-Schmidt norm of the operator of the operator $K_{1,s}$. We estimate
\begin{align}
\| K_{j,s} \|_{\rm HS}^2 = \int dxdy \;  v^2( x-y) \vert \varphi_s (x) \vert^2 \vert \varphi_s (y) \vert^2 \leq \| \varphi_s \|_{L^\infty}^2 \| v \|_{L^2}^2 \| \varphi_s\|_{L^2}^2 \leq \| v \|_{L^2}^2 \| \varphi_s \|_{L^\infty}^2  \; . 
\end{align} 
Since 
\begin{align}
\nabla_x K_{2,s} (x;y) = (\nabla v)(x-y) \varphi_s (x) \varphi_s (y) + v (x-y) (\nabla \varphi_s) (x) \varphi_s (y) 
\end{align}
we find similarly, 
\begin{align}
\| \nabla K_{2,s} \|_{\rm HS}^2 \leq& \int dxdy \; ( \nabla v)^2(x-y)  \; \vert \varphi_s (x) \vert^2 \; \vert \varphi_s (y) \vert^2 \notag \\
&+ \int dxdy \; v^2(x-y) \; \vert \nabla \varphi_s (x) \vert^2 \vert \varphi_s (y) \vert^2 \notag \\
\leq& \| \varphi_s \|_{L^\infty}^2 \bigg( \| \nabla \varphi_s \|_{L^2}^2 \| v \|_{L^2}^2 + \| \nabla v \|_{L^2}^2 \| \varphi_s \|_{L^2}^2 \bigg)  
\end{align}
and the bound for $\| \nabla K_{1,s}\|_{\rm HS}$ follows analogously. For the estimate on the Hilbert-Schmidt norm of $\nabla K_{j,s}$ we proceed in the same way and compute 
\begin{align}
\Delta K_{2,s} (x,y) =& ( \Delta v) (x-y) \varphi_s (x) \varphi_s (y) + 2 (\nabla v)(x-y) ( \nabla \varphi_s) (x) \varphi_s (y) \notag \\
&+ v(x-y) (\Delta \varphi_s)(x) \varphi_s (y) 
\end{align}
leading to the following bound on its Hilbert-Schmidt norm 
\begin{align}
\| \Delta K_{2,s} \|^2_{\rm HS} \leq \| \varphi_s \|_{L^\infty}^2 \bigg( \| \Delta v \|_{L^2}^2 \| \varphi_s \|_{L^2}^2 + \| \nabla v \|_{L^2}^2 \| \nabla \varphi_s \|_{L^2}^2  + \| v \|_{L^2}^2 \| \Delta \varphi_s \|_{L^2}^2 \bigg) 
\end{align}
and, again, similarly for the operator $K_{1,s}$. 

\end{proof}

\subsection{Properties of symplectic Bogoliubov dynamics}
\label{sec:symplbogo}

In this section we study the symplectic Bogoliubov dynamics defined in \eqref{def:bogo}. We recall that the symplectic Bogoliubov dynamics a two-parameter family 
\begin{align}
\Theta_{(t;s)} : L^2( \mathbb{R}^3) \oplus L^2( \mathbb{R}^3) \rightarrow L^2 ( \mathbb{R}^3 ) \oplus L^2 ( \mathbb{R}^3)
\end{align}
given in terms of two bounded linear operators $\gamma_{(t;s)}, \sigma_{(t;s)} : L^2( \mathbb{R}^3) \rightarrow L^2 ( \mathbb{R}^3)$ through  
\begin{align}
\Theta_{(t;s)} = \begin{pmatrix}
\gamma_{(t;s)} & \overline{\sigma_{(t;s)}} \\
\sigma_{(t;s)} & \overline{\gamma_{(t;s)}} 
\end{pmatrix} , \quad \text{with} \quad \gamma_{(t;s)} (x;y)  = \delta (x;y) , \quad \text{and} \quad \sigma_{(t;t)} = 0 \;
\end{align}
satisfying \eqref{def:bogo}. As a consequence of this definition \eqref{def:bogo}, the operator $\sigma_{(t;s)}$ is bounded in Hilbert-Schmidt-norm, while $\gamma_{(t;s)}$ in operator norm, only. However, the operator $\gamma_{(t;s)}$ has a decomposition of the form 
\begin{align}
 \gamma_{(t;s)} (x;y) = \delta (x;y)  + \eta_{(t;s)} (x;y)  \label{def:eta}
\end{align}
for a Hilbert-Schmidt operator $\eta_{(t;s)}$, as we prove in the following Lemma. 

\begin{lemma}
\label{lemma:gamma,sigma}
For $0 \leq s \leq t$, let $\gamma_{(t;s)}, \sigma_{(t;s)}$ denote two operators on $L^2( \mathbb{R}^3)$ satisfying \eqref{def:bogo} and \eqref{def:gamma,sigma}. Then, we have 
\begin{align}
\| \gamma_{(t;s)} \|_{\rm op}^2 \leq& 1 + 2 \int_s^t \| K_\tau \|_{\rm op} e^{\int_\tau^t \| K_{\tau_1} \|_{\rm op} d\tau_1}  \label{eq:bound-gamma-op1}
\end{align} 
and 
\begin{align} 
\| \sigma_{(t;s)} \|_{\rm HS} \leq& 2 \int_s^t \| K_\tau \|_{\rm HS} e^{\int_\tau^t \| K_{\tau_1} \|_{\rm op} d\tau_1} \; .  \label{eq:bound-sigma-HS1}
\end{align}

Moreover, let  $v \in C_0^1( \mathbb{R}^3)$ be nonnegative, spherically symmetric, decreasing, and $t \in \mathbb{R}$, $\varphi_t $ denote the solution to the Hartree equation with initial data $\varphi_0 \in W^{1, \ell } ( \mathbb{R}^3)$ for sufficiently large $\ell >0$ and $H_t, K_t$ be given by \eqref{def:Hs,Ks}. Then, there exists a constant $C>0$ (that depends only on $\| \varphi_0 \|_{W^{1, \ell} ( \mathbb{R}^3)}$) such that 
\begin{align}
\| \gamma_{(t;s)} \|_{\rm op}^2 \leq& 1 + C  \bigg( \frac{1}{(1+s)^2} - \frac{1}{(1+t)^2}\bigg)  \label{eq:bound-gamma-op2} \; . 
\end{align}
Moreover, there exists an operator $\eta_{(t;s)}$ satisfying \eqref{def:eta}, such that 
\begin{align}
\| \sigma_{(t;s)} \|_{\rm HS}, \| \eta_{(t;s)} \|_{\rm HS}  \leq C  \bigg( \frac{1}{(1+s)^{1/2}} - \frac{1}{(t+1)^{1/2}} \bigg) \; .  \label{eq:bounds-sigma-HS2}
\end{align}
Furthermore, assuming that $v  \in C_0^2( \mathbb{R}^3)$, there exists a constant $C>0$ such that
\begin{align}
\| \nabla_1 \sigma_{(t;s)} \|_{\rm HS} , \| \Delta_1 \sigma_{(t;s)} \|_{\rm HS} \leq C  \bigg( \frac{1}{(1+s)^{1/2}} - \frac{1}{(t+1)^{1/2}} \bigg) \; . 
\end{align}
\end{lemma}

\begin{proof}
From \eqref{def:bogo}, we derive the coupled system of differential equations for $\gamma_{(t;s)}, \sigma_{(t;s)}$ 
\begin{align}
\label{eq:pde-gamma-sigma}
i \partial_s \gamma_{(t;s)} =(  \Delta + H_s )  \gamma_{(t;s)} - K_s \sigma_{(t;s)} , \quad
i \partial_s \sigma_{(t;s)} = (  \Delta + H_s ) \sigma_{(t;s)} + \overline{K}_s \sigma_{(t;s)} 
\end{align}
resp. their adjoints 
\begin{align}
i \partial_s \gamma^*_{(t;s)} =- \gamma^*_{(t;s)} (  \Delta + H_s ) + \sigma^*_{(t;s)} \overline{K}_s , \quad i \partial_s \sigma^*_{(t;s)} = - \sigma^*_{(t;s)} (  \Delta + H_s ) - \gamma^*_{(t;s)} K_s \; . 
\end{align}
First, we prove that the operator norm of $\gamma_{(t;s)}$ is bounded. For this we write for any $f \in L^2( \mathbb{R}^3)$ with Duhamel's formula 
\begin{align}
\| \gamma_{(t;s)} f \|_{L^2}^2 =& \| f \|_{L^2}^2 + \int_0^t d\tau \; 2 \Re \partial_\tau \langle \gamma_{(t;\tau)} f, \gamma_{(t;\tau)} f \rangle \notag \\
=& \| f \|_{L^2}^2 + \int_0^t ds \; 2 \Im \langle \sigma_{(t;\tau)} f, K_\tau \gamma_{(t;\tau)} f \rangle  
\end{align}
and estimate with Cauchy-Schwarz inequality; 
\begin{align}
\| \gamma_{(t;s)} f \|_{L^2}^2 \leq &  \| f \|_{L^2}^2 + C \int_0^t d\tau \; \| K_{\tau} \|_{\rm op} \|\sigma_{(t;\tau)} f \|_{L^2} \| \gamma_{(t;\tau} f \|_{L^2} \; . 
\end{align}
With the same ideas, one gets a similar estimate for $\| \sigma_{(t;s)} f \|_{L^2}^2$. Combining now both estimates, for $\| \sigma_{(t;s)} f \|_{L^2}^2$ and $\| \gamma_{(t;s)} f \|_{L^2}^2$ we arrive at 
\begin{align}
\| \gamma_{(t;s)} f \|_{L^2}^2 + \| \sigma_{(t;s)} f \|_{L^2}^2 \leq &  \| f \|_{L^2}^2 + C \int_0^t d\tau \; \| K_{\tau} \|_{\rm op} \|\sigma_{(t;\tau)} f \|_{L^2} \| \gamma_{(t;\tau} f \|_{L^2} 
\end{align}
and conclude with Gronwall's inequality at the desired bound \eqref{eq:bound-gamma-op1} for the operator norms of $\gamma_{(t;s)}$. Then, the second  bound  \eqref{eq:bound-gamma-op2} for the special choices of the operators $K_\tau, H_\tau$ from \eqref{def:Hs,Ks} follows with lemma \ref{lemma:K}.

Next, we aim for bounds on the Hilbert-Schmidt norm of $\sigma_{( t;s)}$. For this we write with Duhamels formula  
\begin{align}
\| \sigma_{(t;s)} \|_{\rm HS}^2 =& \Tr \sigma_{(t;s)} \sigma^*_{(t;s)}  \notag \\
=& 2 \int_s^t d\tau \Tr \big[ K_\tau \gamma_{(t; \tau)} \sigma^*_{(t,\tau)} + \sigma_{(t;\tau)} \gamma^*_{(t;\tau)} \overline{K}_\tau \big] \; d \tau  \label{eq:sigma-HS-step1}
\end{align}
and estimate using the trace's cyclicity 
\begin{align}
\| \sigma_{(t;s)} \|_{\rm HS}^2  \leq 2 \int_s^t d\tau \; \| \gamma_{(t;s)} \|_{\rm op} \| \; \| K_\tau \|_{\rm HS} \| \sigma_{(t;\tau)} \|_{\rm HS}
\end{align}
that leads with the bound of the operator norm of $\gamma_{(t;\tau}$ proven before to 
\begin{align}
\label{eq:gamma-HS-bound}
\| \sigma_{(t;s)} \|_{\rm HS}^2  \leq C \int_s^t d\tau \;  \; \| K_\tau \|_{\rm HS} \| \sigma_{(t;\tau)} \|_{\rm HS}
\end{align}
and the desired bound \eqref{eq:bound-sigma-HS1} for the Hilbert-Schmidt norm of $\sigma_{(t;\tau)}$ follows by a Gronwall type argument. Again, the second bound \eqref{eq:bounds-sigma-HS2} for the special choices of $H_\tau, K_\tau$ from \eqref{def:Hs,Ks} follows with Lemma \ref{lemma:K}.

Next, we aim for bounds on the Hilbert-Schmidt norm of $\eta_{( t;s)} $ that satisfy \eqref{def:eta}. To this end, we define $\eta_{(t;s)}$ as an operator on $L^2( \mathbb{R}^3)$ whose kernel satisfies 
\begin{align}
i \partial_s \eta_{(t;s)} = ( -\Delta + H_s) \gamma_{(t;s)} - K_s \sigma_{(t;s)} , \quad \eta_{(t;t)} =0  \; . \label{def:eta2}
\end{align}
where $\gamma_{(t;s)}$ is the solution to \eqref{eq:pde-gamma-sigma}.  Then, $\eta_{(t;s)}$ satisfies \eqref{def:eta} and its Hilbert-Schmidt norm is, by cyclicity of the trace, given by 
\begin{align}
\| \eta_{(t;s)} \|_{\rm HS}^2 =& \Tr \eta_{(t;s)} \eta^*_{(t;s)} =  \Tr e^{-is \Delta_x }\eta_{(t;s)} \eta^*_{(t;s)}e^{is \Delta_x } 
\end{align}
and we get with Duhamel's formula 
\begin{align}
\| \eta_{(t;s)} \|_{\rm HS}^2 =&  i \int_s^t d\tau \Tr \big[ \big( ( H_\tau - \Delta ) \gamma_{(t;\tau)} + K_\tau \sigma_{(t; \tau)} \big)  \eta^*_{(t,\tau)} + \eta_{(t;\tau)} \big( \sigma^*_{(t;\tau)} \overline{K}_\tau  + \gamma_{(t;\tau)} (H_\tau - \Delta ) \big) \big]  
\end{align}
that we estimate with 
\begin{align}
\| \eta_{(t;s)} \|_{\rm HS}^2 \leq& \int_s^t d\tau \big( \| v * \vert \varphi_\tau \vert^2 \|_{L^2} + \| \widetilde{K}_{1,\tau}  \|_{\rm HS}  \big) \| \gamma_{(t;\tau)} \|_{\rm op} \| \eta_{(t;\tau} \|_{\rm HS} \notag \\
\leq& \int_s^t d\tau \| K_{2,\tau} \|_{\rm op} \| \sigma_{(t;\tau)} \|_{\rm HS} \| \eta_{(t;\tau)} \|_{\rm HS}
\end{align}
Since 
\begin{align}
\| v * \vert \varphi_\tau \vert^2 \|_{L^2} \leq \| v \|_{L^1} \| \vert\varphi_\tau \vert^2 \|_{L^2} \leq \| v \|_{L^1} \| \varphi_\tau \|_{L^\infty} \| \varphi_\tau \|_{L^2} 
\end{align}
we find with Proposition \ref{prop:dispersive}, Lemma \ref{lemma:K} and the estimates on $\| \sigma_{(t;\tau)} \|_{\rm HS}$ and $\| \gamma_{(t;\tau)} \|_{\rm op}$ proven before, that 
\begin{align}
\label{eq:eta-bound-HS}
\| \eta_{(t;s)} \|_{\rm HS}^2  \leq C \int_s^t d \tau \frac{1}{(1 + \vert \tau \vert)^{3/2}} \| \eta_{(t;\tau)} \|_{\rm HS} 
\end{align}
that implies the desired bound by Gronwall's argument. 

Next, we show bounds on the Hilbert-Schmidt norm of $\Delta_1 \sigma_{(t;s)}$. For this, we proceed similarly as before and write 
\begin{align}
\| \nabla_1 \sigma_{(t;s)} \|_{\rm HS}^2 =& \Tr (\nabla \sigma_{(t;s)}) (\nabla \sigma_{(t;s)})^* \notag \\
=& i \int_s^t d\tau \; \Tr \bigg[ \nabla \big( \Delta + H_\tau \big) \sigma_{(t;\tau)} \sigma_{(t;\tau)}^* \nabla -   \nabla \sigma_{(t;\tau)} \sigma_{(t;\tau)}^* \big( \Delta + H_\tau \big) \nabla \bigg] \notag \\
&+ i\int_s^t d\tau \; \Tr \bigg[ \nabla K_\tau \gamma_{(t;\tau)} \sigma_{(t;\tau)}^* \nabla -  \nabla \sigma_{(t;\tau)} \gamma_{(t;\tau)}^* K_\tau  \nabla \bigg] \notag \\
=& i \int_s^t d\tau \; \Tr \big[ \nabla, v* \vert \varphi_\tau\vert^2 \big] \sigma_{(t;\tau)} \sigma_{(t;\tau)}^* \nabla \notag \\
&+ i \int_s^t d\tau \; \Tr \bigg[ \nabla \widetilde{K}_{1,\tau} \sigma_{(t;\tau)} \sigma_{(t;\tau)}^* \nabla - \nabla \sigma_{(t;\tau)} \sigma_{(t;\tau)}^*\widetilde{K}_{1,\tau}  \nabla  \bigg] \notag \\
&+ i \int_s^t d\tau \; \Tr \bigg[ \nabla \widetilde{K}_{2,\tau} \gamma_{(t;\tau)} \sigma_{(t;\tau)}^* \nabla - \nabla \sigma_{(t;\tau)} \gamma_{(t;\tau)}^*\widetilde{K}_{2,\tau}  \nabla  \bigg] \notag \\
\end{align}
Since 
\begin{align}
\big[ \nabla, (v* \vert \varphi_\tau \vert^2  \big] = 2 \Re v* ( \varphi_\tau \nabla \varphi_\tau) 
\end{align}
we arrive at the upper bound 
\begin{align}
\| \nabla_1 \sigma_{(t;s)} \|_{\rm HS}^2 \leq& 4 \int_s^t d\tau \;  \| v* ( \varphi_\tau \nabla \varphi_\tau)  \|_{L^\infty} \| \sigma_{(t;\tau)} \|_{\rm HS} \| \nabla \sigma_{t, \tau} \|_{\rm HS } \notag \\
&+ 2 \int_s^t d\tau \; \| \nabla \widetilde{K}_{1, \tau} \sigma_{(t;\tau)} \|_{\rm HS} \| \nabla \sigma_{(t;\tau)} \|_{\rm HS} \notag \\
&+ 2 \int_s^t d\tau \; \| \nabla \widetilde{K}_{2, \tau} \gamma_{(t;\tau)} \|_{\rm HS} \| \nabla \sigma_{(t;\tau)} \|_{\rm HS} \notag \\
\end{align}
that we can further bound, using that by Minkowski inequality and Proposition \ref{prop:dispersive},  
\begin{align} 
\| v * ( \varphi_\tau \nabla \varphi_\tau ) \|_{L^\infty} =& \| v \|_{L^2} \|  \varphi_\tau \nabla \varphi_\tau \|_{L^2} \leq \| v \|_{L^2} \| \varphi_\tau \|_{L^\infty} \| \nabla \varphi_\tau \|_{L^2} \leq \frac{C \| v \|_{L^2}}{(1+ \vert \tau \vert)^{3/2}}
\end{align}
as well as $\| \nabla \widetilde{K}_{1, \tau} \sigma_{(t;\tau)} \|_{\rm HS} \leq \| \nabla \widetilde{K}_{1, \tau} \|_{\rm HS}  \| \sigma_{(t;\tau)} \|_{\rm HS}$ together with the estimates on $\widetilde{K}_{j, \tau}$ from Lemma \ref{lemma:K}, by 
\begin{align}
\| \nabla_1 \sigma_{(t;s)} \|^2_{\rm HS} \leq& C \| v \|_{H^2}\int_s^t d\tau \;  \frac{1}{(1 + \vert \tau \vert)^{3/2}} \| \sigma_{(t;\tau)} \|_{\rm HS} \| \nabla \sigma_{t, \tau} \|_{\rm HS } \notag \\
&+ C \| v \|_{H^2}\int_s^t d\tau \;  \frac{1}{(1 + \vert \tau \vert)^{3/2}} \| \gamma_{(t,\tau)} \|_{\rm op} \| \nabla \sigma_{(t;\tau)} \|_{\rm HS} \; . 
\end{align}
Furthermore, by the bounds  $\| \gamma_{(t,\tau)} \|_{\rm HS}, \| \sigma_{(t,\tau)} \|_{\rm HS} \leq C $ for all $t,\tau \in \mathbb{R}$, just proven before,  we arrive at 
\begin{align}
\label{eq:nablagamma-bound-HS}
\| \nabla_1 \sigma_{(t;s)} \|^2_{\rm HS} \leq& C \int_s^t d\tau \;  \frac{1}{(1 + \vert \tau \vert)^{3/2}}  \| \nabla \sigma_{t, \tau} \|_{\rm HS } 
\end{align}
where the constant depends on $\| v \|_{H^2}$. Finally, we conclude by a Gronwall type argument 
\begin{align}
\| \nabla_1 \sigma_{(t;s)} \|_{\rm HS}^2 \leq C  \bigg( \frac{1}{(1+ \vert s \vert)^{1/2}} - \frac{1}{(1+ \vert t \vert)^{1/2}} \bigg) \; . 
\end{align}
As a last step of this proof, we derive a bound on the Hilbert-Schmidt norm for $\| \Delta \sigma_{(t;s} \|_{\rm HS}$ for which we follow the same strategy as before and  compute first. 
\begin{align}
\| \Delta_1 \sigma_{(t;s)} \|_{\rm HS}^2 =& \Tr (\Delta \sigma_{(t;s)}) (\Delta \sigma_{(t;s)})^* \notag \\
=& i \int_s^t d\tau \; \Tr \big[ \Delta, v* \vert \varphi_\tau\vert^2 \big] \sigma_{(t;\tau)} \sigma_{(t;\tau)}^* \Delta \notag \\
&+ i \int_s^t d\tau \; \Tr \bigg[ \Delta \widetilde{K}_{1,\tau} \sigma_{(t;\tau)} \sigma_{(t;\tau)}^* \Delta - \Delta \sigma_{(t;\tau)} \sigma_{(t;\tau)}^*\widetilde{K}_{1,\tau}  \Delta  \bigg] \notag \\
&+ i \int_s^t d\tau \; \Tr \bigg[ \Delta \widetilde{K}_{2,\tau} \gamma_{(t;\tau)} \sigma_{(t;\tau)}^* \Delta - \Delta \sigma_{(t;\tau)} \gamma_{(t;\tau)}^*\widetilde{K}_{2,\tau}  \Delta  \bigg] 
\end{align}
Since 
\begin{align}
\big[ \Delta, v* \vert \varphi_\tau \vert^2 \big] =& 2  \big[ \nabla, v* \vert \varphi_\tau \vert^2 \big] \nabla + \big[ \nabla, \big[ \nabla, v* \vert \varphi_\tau \vert^2 \big] \big] \notag \\
=& 4 \Re ( \nabla v)* \vert \varphi_\tau \vert^2  \nabla + ( \Delta v) * \vert \varphi_\tau \vert^2 
\end{align}
we arrive at 
\begin{align}
\| \Delta_1 \sigma_{(t;s)}  \|_{\rm HS}^2 \leq& C \int_s^t d\tau \; \| ( \nabla v ) * \vert \varphi_\tau \vert^2 \|_{L^\infty} \| \nabla \sigma_{(t;\tau)} \|_{\rm HS} \| \Delta_1 \sigma_{(t;\tau)} \|_{\rm HS} \notag \\
&+ C \int_s^t d\tau \;  \| (\Delta v) * \vert \varphi_\tau \vert^2 \|_{L^\infty} \| \sigma_{(t;\tau)} \|_{\rm HS} \| \Delta_1 \sigma_{(t;\tau)} \|_{\rm HS} \notag \\
&+ C \int_s^t d\tau \| \Delta \widetilde{K}_{1,\tau} \|_{\rm HS} \| \sigma_{(t;\tau)} \|_{\rm HS} \| \Delta_1 \sigma_{(t;\tau)} \|_{\rm HS} \notag \\
&+ C \int_s^t d\tau \| \Delta \widetilde{K}_{2,\tau} \|_{\rm HS} \| \gamma_{(t;\tau)} \|_{\rm  op} \| \Delta_1 \sigma_{(t;\tau)} \|_{\rm HS} 
\end{align}
and we conclude with $\| ( \nabla v)  * \vert \varphi_\tau \vert^2 \|_{L^\infty},\| ( \Delta v)  * \vert \varphi_\tau \vert^2 \|_{L^\infty} \leq C \| v \|_{H^2} \| \varphi_\tau \|_{L^\infty} \| \varphi_\tau \|_{L^2} \leq C \| v \|_{H^2} (1 + \vert \tau \vert)^{-3/2}$ by Proposition \ref{prop:dispersive}, Lemma \ref{lemma:K} and the estimates on $\| \sigma_{(t; \tau} \|_{\rm HS}, \|\nabla \sigma_{(t;\tau)} \|_{\rm HS}$ proven before, that 
\begin{align}
\| \Delta_1 \sigma_{(t;s)}  \|_{\rm HS}^2 \leq& C \int_s^t d\tau \frac{1}{(1 + \vert \tau \vert)^{3/2}} \| \Delta_1 \sigma_{(t;\tau} \|_{\rm HS} \label{eq:deltagamma-bound-HS}
\end{align}
where the constant depends on $\| v \|_{H^2}$, finally leading by a Gronwall type argument to the desired bound. 
\end{proof}

Summarizing, Lemma \ref{lemma:K} shows that the Hilbert-Schmidt norm of the operators $ \sigma_{(t;s)}$ is for all times $t,s \in \mathbb{R}$ bounded independent of $t$. In the next section we prove Theorem \ref{thm:bogo-disp} and show that a suitably defined norm of $\sigma_{(t;s)}$ actually disperses in time.

\subsection{Proof of Theorem \ref{thm:bogo-disp}} 
\label{sec:proofthmdisp}

In this section we prove Theorem \ref{thm:bogo-disp}.

\begin{proof}[Proof of Theorem \ref{thm:bogo-disp}] 
The proof is based on techniques derived for dispersive estimates for solutions of non-linear Schnr\"odinger equations and relies on dispersive behaviour of the solution of the Hartree equation $\varphi_t$ (Proposition \ref{prop:dispersive}) that the non-linearity of the coupled system of differential equations of $(\gamma_{(t;s)}, \sigma_{(t;s)})$ in \eqref{eq:pde-gamma-sigma} are built, namely 
\begin{align}
\label{eq:pde-gamma-2}
i \partial_s \gamma_{(t;s)}  =(  \Delta + H_s )  \gamma_{(t;s)} - K_s \sigma_{(t;s)} , \quad
i \partial_s \sigma_{(t;s)} = ( \Delta + H_s ) \sigma_{(t;s)} + \overline{K}_s \gamma_{(t;s)}  
\end{align}
where we recall that $H_s := v * \vert \varphi_s \vert^2 + \widetilde{K}_{1,s}$ and $K_s := \widetilde{K}_{2,s}$. We fix $t \in \mathbb{R}$ and write $\sigma_{(t;s)}$ with Duhamel's formula, recalling that $\sigma_{(t;t)} = 0$,  
\begin{align}
\sigma_{(t;s)} (x;y)  =& - i \int_s^t d\tau \; e^{-i(t-\tau) \Delta_x}  (v * \vert \varphi_\tau \vert^2 )(x) \;  \sigma_{(t;\tau)} (x;y)  \notag \\
&- i \int_s^t d\tau \; e^{-i(t-\tau) \Delta_x} \int dz \;  \widetilde{K}_{1,\tau} (x;z) \sigma_{(t;\tau)} (z,y)    \notag \\
&+ i \int_s^t  d\tau \; e^{-i(t-\tau) \Delta_x}  \int  dz  \; \widetilde{K}_{2,\tau}(x;z) \gamma_{(t;\tau)} (z;y) \; , 
\end{align}
leading to 
\begin{align}
\| \sigma_{(t;s)} \|_{L^\infty \times L^2} \leq&  \int_s^t d\tau \; \| e^{-i(t-\tau) \Delta_1}  (v * \vert \varphi_\tau \vert^2 ) \;  \sigma_{(t;\tau)} \|_{L^\infty \times L^2} \notag \\
&+ \int_s^t d\tau \;  \| e^{-i (t-\tau) \Delta_1} \widetilde{K}_{1,\tau} \sigma_{(t;\tau)} \|_{L^\infty \times L^2}  \notag \\
&+ \int_s^t d\tau \; \| e^{-i(t-\tau) \Delta_1} \widetilde{K}_{2,\tau} \gamma_{(t;\tau)} \|_{L^\infty \times L^2} \;  .  \label{eq:estimate-sigma-eins}
\end{align}
The idea is to estimate all the three terms of the r.h.s. using ideas built on arguments from dispersive estimates (see for example \cite{DL,GM1}). For this we recall that in order to prove Theorem \ref{thm:bogo-disp}, we need to show that 
\begin{align}
 \| \sigma_{(t;s)} \|_{L^\infty \times L^2} \leq \frac{C}{ ( 1+ \vert t -s \vert)^{3/2}}\label{eq:bound-eins}
\end{align}
for a constant $C>0$ independent of $t,s$. In fact, instead of \eqref{eq:bound-eins}, we  prove that 
 \begin{align}
\| \sigma_{(t;s)} \|_{L^\infty \times L^2}  & + \frac{\varepsilon}{( 1+ \vert t -s \vert)^{3/2}} M(t_0;t)   \leq \frac{C}{( 1 + \vert t-s\vert)^{3/2}} \label{eq:bound-zwei}
\end{align} 
where $t_0 = \max\lbrace s,t -1 \rbrace$
\begin{align}
M(t_0;t) :=  \sup_{\tau \in [t_0,t]} \bigg( \| \sigma_{(t;\tau)} \|_{\rm HS}  + \| \nabla_1 \sigma_{(t;\tau)} \|_{\rm HS} + \|\Delta_1 \sigma_{(t;\tau)} \|_{\rm HS}  + \| \eta_{(t; \tau)} \|_{\rm HS} + 1\bigg) 
\end{align}
for some $\varepsilon, C>0$. And since \eqref{eq:bound-eins} is obvious for small time, it is sufficient to prove this bound for $t\ge T \ge 2$ with a large constant $T>2$ depending only on $\varphi_0$ and $v$, which will be specified later. 

The reason for the auxiliary second term in \eqref{eq:bound-zwei}, compared to \eqref{eq:bound-eins}, is that in order to estimate the r.h.s. of \eqref{eq:estimate-sigma-eins} we separate the integral over $\tau \in [s,t]$ into two regimes: For this we fix $t_0 = \max \lbrace s, t-1 \rbrace$. While in the first regime, for $\tau \in [s, t_0]$, we show that the integral in fact decays as $(1+ \vert t-s\vert)^{-3}$ as desired;  for the second region $\tau \in [t_0, t]$, we prove an upper bound of the integral in terms of $ M(t_0, t)$. Thus, in order to close the argument, we need the second line of the r.h.s. \eqref{eq:bound-zwei}.

We start with the first term of the r.h.s. of \eqref{eq:estimate-sigma-eins} and consider the first regime $\tau \in [s,t_0]$. By definition of the $L^\infty\times L^2$-norm, we have 
\begin{align}
\| e^{-i(t-\tau) \Delta_1}  & (v * \vert \varphi_s \vert^2  ) \sigma_{(t;\tau)} \|_{L^\infty \times L^2}^2 \notag \\
=& \sup_{x \in \mathbb{R}^3} \int dy \; \big\vert e^{-i(t-\tau) \Delta_x} (v* \vert \varphi_\tau \vert^2 ) (x) \sigma_{(t;\tau)} (x;y) \big\vert^2 \notag \\
 =&  \int dy \sup_{x \in \mathbb{R}^3} \; \big\vert e^{-i(t- \tau) \Delta_x} (v* \vert \varphi_\tau \big\vert^2 )(x) \sigma_{(t;\tau)} (x;y) \big\vert^2 \; .  \label{eq:sup-ex}
\end{align}
In fact, we can switch the supremum with the integral here, since 
\begin{align}
\sup_{x \in \mathbb{R}^3} \; & \big\vert e^{-i(t-\tau) \Delta_x}  (v* \vert \varphi_\tau \vert^2 )(x) \sigma_{(t;\tau)} (x;y) \big\vert^2 \notag \\
=& \bigg(\sup_{x \in \mathbb{R}^3} \; \big\vert e^{-i(t-\tau) \Delta_x}  (v* \vert \varphi_\tau \vert^2 )(x) \sigma_{(t;\tau)} (x;y) \big\vert\bigg)^2
\end{align}
is integrable with respect to the variable $y$: This is in fact a consequence of the embedding of $L^\infty$ in $
H^2$ which leads to the bound 
\begin{align}
\sup_{x \in \mathbb{R}^3} \; & \big\vert e^{-i(t-\tau) \Delta_x}  (v* \vert \varphi_\tau \vert^2 ) \sigma_{(t;\tau)} (x;y) \big\vert  \notag \\
\leq& \bigg( \int dx \big\vert ( -\Delta_x + 1) e^{-i(t- \tau) \Delta_x} e^{i (1-s) \Delta_x} (v* \vert \varphi_\tau \vert^2 )(x)  \sigma_{(t;\tau)} (x;y) \big\vert^2 \bigg)^{1/2} \notag \\
=& \bigg( \int dx \big\vert ( -\Delta_x + 1)  (v* \vert \varphi_\tau \vert^2 )(x)  \sigma_{(t;\tau)} (x;y) \big\vert^2 \bigg)^{1/2} \; .
\end{align}
And the r.h.s. is bounded, as by the chain rule and 
\begin{align}
\| v * \vert \varphi_\tau \vert^2 \|_{\infty} \leq \| v \|_{L^2}   \| \varphi_\tau \|_4^2  \leq C \| v \|_{L^2} \| \varphi_\tau \|_{L^\infty} \| \varphi_\tau \|_{L^2} 
\end{align}
resp. 
\begin{align}
\| (\nabla v) * \vert \varphi_\tau \vert^2 \|_{\infty}, \; \| (\Delta v) * \vert \varphi_\tau \vert^2 \|_{\infty} \leq \| v \|_{H^2}  \| \varphi_\tau \|_{L^\infty} \| \varphi_\tau \|_{L^2}
\end{align}
we have 
\begin{align}
\sup_{x \in \mathbb{R}^3} \; & \big\vert e^{-i(t-\tau) \Delta_x}  (v* \vert \varphi_\tau \vert^2 ) \sigma_{(t;\tau)} (x;y) \big\vert  \notag \\
\leq& \bigg( \int dx \big\vert ( -\Delta_x + 1) e^{-i(t- \tau) \Delta_x} e^{i (1-s) \Delta_x} (v* \vert \varphi_\tau \vert^2 )(x)  \sigma_{(t;\tau)} (x;y) \big\vert^2 \bigg)^{1/2} \notag \\
=& C \| v \|_{H^2} \| \varphi_\tau \|_{L^\infty} \| \varphi_\tau \|_{L^2} \bigg( \| \sigma_{(t;\tau)} ( \cdot, y ) \|_{2} + \| \nabla_1 \sigma_{(t;\tau)} ( \cdot, y ) \|_{2}+  \| \Delta_1 \sigma_{(t;\tau)} ( \cdot, y ) \|_{2}  \bigg) 
\end{align}
and therefore 
\begin{align}
\int dy  &  \sup_{x \in \mathbb{R}^3} \; \big\vert e^{-i(t-\tau) \Delta_x}  (v* \vert \varphi_\tau \vert^2 ) \sigma_{(t;\tau)} (x;y) \big\vert^2\notag \\
\leq& C \| v \|_{H^2}^2 \| \varphi_\tau \|_{L^\infty}^2 \bigg( \| \sigma_{(t;\tau)}\|_{\rm HS}^2 + \| \nabla_1 \sigma_{(t;\tau)} \|_{\rm HS}^2+  \| \Delta_1 \sigma_{(t;\tau)}^2 \|_{\rm HS}^2  \bigg)  < \infty  \label{eq:first-term-h2-bound}
\end{align}
which is bounded by Lemma \ref{lemma:gamma,sigma}. 

Thus, going back to \eqref{eq:sup-ex}, the goal is to estimate the integrand 
\begin{align}
\| e^{-i(t-\tau) \Delta_1}  & (v * \vert \varphi_\tau \vert^2  ) \sigma_{(t;\tau)} \|_{L^\infty \times L^2}^2 \notag \\
 =&  \int dy \bigg(\sup_{x \in \mathbb{R}^3} \; \big\vert e^{-i(t- \tau) \Delta_x} (v* \vert \varphi_\tau \big\vert^2 )(x) \sigma_{(t;\tau)} (x;y) \big\vert\bigg)^2 \; . 
\end{align}
in the regime $\tau \in [s, t_0]$. For this, we use standard dispersive estimates w.r.t. the $x$-variable leading to 
\begin{align}
\sup_{x \in \mathbb{R}^3} & \; \big\vert e^{-i(t- \tau) \Delta_x} (v* \vert \varphi_\tau \big\vert^2 )(x) \sigma_{(t;\tau)} (x;y) \big\vert \notag \\
\leq& \frac{1}{\vert t-\tau \vert^{3/2}}\int  dx \; \vert  (v* \vert \varphi_\tau \vert^2) (x) \sigma_{(t;\tau)} (x;y) \vert \notag \\ 
\leq& \frac{1}{\vert t-\tau \vert^{3/2}} \| v * \vert \varphi_\tau \vert^2 \|_{L^2} \| \sigma_{(t; \tau)} ( \cdot,y) \|_{L^2}  \; . 
\end{align}
Since 
\begin{align}
\| v * \vert \varphi_\tau \vert^2 \|_{L^2} \leq \| v \|_{1} \| \vert \varphi_\tau\vert^2 \|_{2} \leq \| v \|_{1} \| \varphi_\tau \|_{L^\infty} \| \varphi_\tau \|_{L^2}  
\end{align}
we furthermore get 
\begin{align}
\| e^{-i(t-\tau) \Delta_1}   (v * \vert \varphi_s \vert^2  ) \sigma_{(t;\tau)} \|_{L^\infty \times L^2}^2 
 \leq   \frac{1}{\vert t-\tau \vert^{3}} \| v \|_{L^1}^2 \| \varphi_\tau \|_{L^\infty}^2 \| \sigma_{(t;\tau)} \|_{\rm HS}^2 
\end{align}
and, since the $L^\infty$-norm of the solution $\varphi_\tau$ of the Hartree equation disperses (see Proposition \ref{prop:dispersive}), and $\| \sigma_{(t;\tau)}\|_{\rm HS} \leq C$ for all $t,\tau \in \mathbb{R}$ by Lemma \ref{lemma:gamma,sigma}, we thus get 
\begin{align}
\| e^{-i(t-\tau) \Delta_1}   (v * \vert \varphi_s \vert^2  ) \sigma_{(t;\tau)} \|_{L^\infty \times L^2} 
 \leq   \frac{1}{\vert t-\tau \vert^{3/2}} \frac{1}{(1 + \vert \tau \vert)^{3/2}}\| v \|_{L^1}  \; . 
\end{align} 
Summarizing, the first term of the r.h.s. of \eqref{eq:estimate-sigma-eins} is for $\tau \in [s, t_0] $ bounded by 
\begin{align}
\int_s^{t_0} d\tau \; \| e^{-i(t-\tau) \Delta_1} (v * \vert \varphi_\tau \vert^2  ) \sigma_{(t;\tau)} \|_{L^\infty \times L^2}^2 \leq C \int_s^{t_0} d\tau \frac{1}{\vert t-\tau \vert^{3/2}} \frac{1}{(1+ \vert \tau \vert)^{3/2}}  \; . \label{eq:time-integral1}
\end{align}
To estimate the integral in time, we split the integration in two regions: First we consider the regime where $\tau  -s +1\leq \frac{t -s +1}{2} $ where we have the inequality $\vert t - \tau \vert = \vert t -s +1 - (\tau - s +1)\vert \geq \vert t+1\vert/2$ and therefore 
\begin{align}
\frac{1}{\vert t -\tau \vert^{3/2}} \leq \frac{8}{\vert t-s+ 1\vert^{3/2} } 
\end{align}
leading with $t_1 =(t-s+1)/2 + s-1$ to 
\begin{align}
\int_{s}^{t_1} d\tau \frac{1}{\vert t -\tau\vert^{3/2}} \frac{1}{(1+ \vert \tau \vert)^{3/2}} &\leq \frac{8}{\vert t-s+ 1\vert^{3/2} }  \int_{s}^{t_1} d\tau \frac{1}{(1+ \vert \tau \vert)^{3/2}} \notag \\
&\leq \frac{C}{\vert t-s+ 1\vert^{3/2} }  \; . 
\end{align}
Second, for $\tau  - s + 1 > (t-s+1)/2 $, we have since $0 < s < \tau $ the estimate $\frac{1}{(1+ \vert \tau \vert)^{3/2}} = \frac{1}{(1+  \tau )^{3/2}} \leq \frac{1}{(\tau - s +1)^{3/2}} \leq \frac{8}{\vert t-s +1 \vert^{3/2}}$ leading to 
\begin{align}
\int_{t_1}^{t_0} d\tau \frac{1}{\vert t-\tau \vert^3} \frac{1}{(1+ \vert \tau \vert)^{3/2}} \leq& \frac{8}{\vert t -s + 1\vert^{3/2} }  \int_{t_1}^{t_0} d\tau \frac{1}{\vert t-\tau \vert^{3/2}}  \; . 
\end{align}
We recall that we set $t_1 =(t-s+1)/2 + s-1 $ and $t_0 = \max \lbrace s,t -1 \rbrace $ and thus the integral of the r.h.s. converges and we arrive at 
\begin{align}
\int_{t_1}^{t_0} d\tau \frac{1}{\vert t-\tau \vert^3} \frac{1}{(1+ \vert \tau \vert)^{3/2}} \leq& \frac{C}{\vert t -s + 1\vert^{3/2} } \; .  \label{eq:time-integral2}
\end{align} 
Summing up, the estimates \eqref{eq:time-integral1} and \eqref{eq:time-integral2} show, that, in the regime $\tau \in [s,t_0]$ the first term of the r.h.s. of \eqref{eq:estimate-sigma-eins} is bounded by 
\begin{align}
\int_s^{t_0} \| e^{-i(t-\tau) \Delta_1} (v * \vert \varphi_\tau \vert^2  ) \sigma_{(t;\tau)} \|_{L^\infty \times L^2}  \leq \frac{C}{\vert t -s + 1\vert^{3/2}}\label{eq:estimate-sigma-term1-tau1} 
\end{align}
that is of the desired form.
 
Next, we consider the first term of the r.h.s. of \eqref{eq:time-integral1} in the regime $\tau \in [t_0,t]$. In this regime we use the bound derived at the beginning in \eqref{eq:first-term-h2-bound} based on the embedding of $L^\infty$ in $H^2$. In fact, from \eqref{eq:first-term-h2-bound} we have 
\begin{align}
\int  & dy \sup_{x \in \mathbb{R}^3} \big\vert e^{-i(t-\tau) \Delta_x} ( v* \vert \varphi_\tau \vert^2) (x) \sigma_{(t;\tau)} (x;y) \big\vert^2 \notag \\
\leq& C \| v \|_{H^2} \| \varphi_\tau\|_{L^\infty}^2 \bigg( \| \sigma_{(t;\tau)} \|_{\rm HS}^2 
+ \| \nabla_1 \sigma_{(t,\tau)} \|_{\rm HS}^2 + \| \Delta_1 \sigma_{(t;\tau)} \|_{\rm HS}^2 \bigg)  \; 
\end{align}
and therefore, using the dispersive estimates for the $L^\infty$-norm of $\varphi_\tau$ from Proposition \ref{prop:dispersive}, we get 
\begin{align}
\int_{t_0}^{t} d\tau \; & \| e^{-i(t-\tau) \Delta_1} ( v* \vert \varphi_\tau \vert^2 ) \sigma_{(t;\tau)} \|_{L^\infty \times L^2} \notag \\
&\leq \frac{ C \| v \|_{H^2}}{(\vert t_0 \vert + 1)^{1/2} } \sup_{\tau \in [t_0,t]} \bigg( \| \sigma_{(t;\tau)} \|_{\rm HS} 
+ \| \nabla_1 \sigma_{(t,\tau)} \|_{\rm HS} + \| \Delta_1 \sigma_{(t;\tau)} \|_{\rm HS} \bigg)   \label{eq:estimate-sigma-term1-tau2}
\end{align}
Summing up \eqref{eq:estimate-sigma-term1-tau1} and \eqref{eq:estimate-sigma-term1-tau2} imply that the first term of the r.h.s. of \eqref{eq:estimate-sigma-eins} is bounded by 
\begin{align}
 \int_{s}^t d\tau  &  \| e^{-i(t-\tau) \Delta_1} ( v* \vert \varphi_\tau \vert^2 ) \sigma_{(t;\tau)} \|_{L^\infty \times L^2}  \notag \\
\leq& \frac{C_1}{ ( 1 + \vert t-s\vert)^{3/2}}  + \frac{ C_2 }{(\vert t_0 \vert + 1) }  \sup_{\tau \in [t_0,t]}  \bigg( \| \sigma_{(t;\tau)} \|_{\rm HS} 
+ \| \nabla_1 \sigma_{(t,\tau)} \|_{\rm HS} + \| \Delta_1 \sigma_{(t;\tau)} \|_{\rm HS} \bigg)   \label{eq:estimate-sigma-term1-final}
\end{align}
for some constants $C_1,C_2 >0$. 

Next, we prove similar bounds for the second term of the r.h.s. of \eqref{eq:estimate-sigma-eins}. Again, we start with the regime $\tau \in [s,t_0]$ with $t_0 = \max \lbrace s, t-1 \rbrace$ and write 
\begin{align}
\| e^{-i(t-\tau) \Delta_1} \widetilde{K}_{1,\tau} \sigma_{(t; \tau )} \|_{L^\infty \times L^2}^2 =& \sup_{x \in \mathbb{R}^3} \int dy \; \bigg\vert \int dz \; e^{-i(t-\tau) \Delta_x}\widetilde{K}_{1,\tau} (x; z) \sigma_{(t;\tau)} (z;y) \bigg\vert^2    \notag \\
=&  \int dy \sup_{x \in \mathbb{R}^3} \; \bigg\vert \int dz \; e^{-i(t-\tau) \Delta_x}\widetilde{K}_{1,\tau} (x; z) \sigma_{(t;\tau)} (z;y) \bigg\vert^2 \; .\label{eq:sup-int-2}
\end{align}
In fact, we can switch the supremum with the integration here since, by the embedding of $L^\infty$ in $H^2$ and Cauchy Schwarz' inequality, we have the bound 
\begin{align}
\sup_{x \in \mathbb{R}^3} \;&  \bigg\vert \int dz \; e^{-i(t-\tau) \Delta_x}\widetilde{K}_{1,\tau} (x; y) \sigma_{(t;\tau)} (y;z) \bigg\vert^2  \notag \\
=& \bigg( \sup_{x \in \mathbb{R}^3}  \big\vert \int dz \; e^{-i(t-\tau) \Delta_x}\widetilde{K}_{1,\tau} (x; z) \sigma_{(t;\tau)} (z;y)\big\vert  \bigg)^2 \notag \\
\leq& C  \int dx \bigg\vert \int dz \; (-\Delta_x +1) \; e^{-i(t-\tau) \Delta_x}\widetilde{K}_{1,\tau} (x; z) \sigma_{(t;\tau)} (z;y) \bigg\vert^2 \notag \\
\leq& C \| \sigma_{(t;\tau)} (\cdot, y) \|_{L^2}^2   \int dx \| (  (-\Delta_x +1) \; e^{-i(t-\tau) \Delta_x}\widetilde{K}_{1,\tau} (x; \cdot) \|_{L^2}^2 \; . 
\end{align}
Since $\|  (-\Delta_x +1) \; e^{-i(t-\tau) \Delta_x}\widetilde{K}_{1,\tau} (x; \cdot) \|_{L^2}^2 = \|  (-\Delta_x +1) \widetilde{K}_{1,\tau} (x; \cdot) \|_{L^2}^2$ we arrive with Lemma \ref{lemma:K} at 
\begin{align}
\sup_{x \in \mathbb{R}^3} \;&  \bigg\vert \int dz \; e^{-i(t-\tau) \Delta_x}\widetilde{K}_{1,\tau} (x; y) \sigma_{(t;\tau)} (y;z) \bigg\vert^2  \notag \\
\leq& C \| \sigma_{(t;\tau)} (\cdot, y) \|_{L^2}^2   \|   (-\Delta +1) \; e^{-i(t-\tau) \Delta_x}\widetilde{K}_{1,\tau}  \|_{\rm HS}^2 \notag \\
&\leq \frac{C \| v \|_{H^2}}{(1+ \vert \tau \vert)^{3/2}} \| \sigma_{(t;\tau)} (\cdot, y) \|_{L^2}^2 
\end{align}
which is in fact an integrable in $y$, as 
\begin{align}
\int dy  \sup_{x \in \mathbb{R}^3} \;&  \bigg\vert \int dz \; e^{-i(t-\tau) \Delta_x}\widetilde{K}_{1,\tau} (x; y) \sigma_{(t;\tau)} (y;z) \bigg\vert^2 \notag \\
&\leq \frac{C \| v \|_{H^2}}{(1+ \vert \tau \vert)^{3/2}} \| \sigma_{(t;\tau )} \|_{\rm HS}^2 < \infty \label{eq:estimate-sigma-term2-anfang}
\end{align} 
and thus \eqref{eq:sup-int-2} holds by Lemma \ref{lemma:gamma,sigma}.  

Thus, going back to \eqref{eq:sup-int-2}, we estimate for $\tau \in [s,t_0]$ with $t_0 = \max \lbrace s, t-1 \rbrace $ as for the first term using dispersive estimates for the Laplace operator and Cauchy Schwarz' inequality 
\begin{align}
\int dy  \sup_{x \in \mathbb{R}^3} \;&  \bigg\vert \int dz \; e^{-i(t-\tau) \Delta_x}\widetilde{K}_{1,\tau} (x; y) \sigma_{(t;\tau)} (y;z) \bigg\vert^2 \notag \\
\leq& \frac{C}{(1 + \vert \t-\tau\vert)^{3}}\int dy \bigg( \int dx  \bigg\vert \int dz \; \widetilde{K}_{1,\tau} (x; y) \sigma_{(t;\tau)} (y;z)  \bigg\vert \bigg)^2 \notag \\
\leq& \frac{C}{(1 + \vert t- \tau \vert)^3} \int dy \bigg( \int dx \| \widetilde{K}_{1,\tau} (x, \cdot ) \|_{L^2} \| \sigma_{(t; \tau)} (y, \cdot ) \|_{L^2} \bigg)^2  \notag \\
=&  \frac{C}{(1 + \vert t- \tau \vert)^3} \| \widetilde{K}_{1, \tau} \|_{\rm HS}^2 \| \sigma_{(t;\tau )} \|_{L^2}^2  \; . 
\end{align}
Since $\| \widetilde{K}_{1, \tau} \|_{\rm HS}^2  \leq C \| v \|_{L^2} (1 + \vert \tau \vert)^{-3}$ from Lemma \ref{lemma:K} and $\| \sigma_{(t;\tau)} \|_{L^2}^2 \leq C$ from Lemma \ref{lemma:gamma,sigma}, we arrive at 
\begin{align}
\int_s^{t_0} d\tau \bigg( \int dy  \sup_{x \in \mathbb{R}^3} \;&  \bigg\vert \int dz \; e^{-i(t-\tau) \Delta_x}\widetilde{K}_{1,\tau} (x; y) \sigma_{(t;\tau)} (y;z) \bigg\vert^2 \bigg)^{1/2}\notag \\
=&  C \int_s^{t_0} d\tau \frac{1}{(1 + \vert t- \tau \vert)^{3/2}} \frac{1}{(1 + \vert \tau \vert)^{3/2}} \leq \frac{C}{ (1 + \vert t-s\vert)^{3/2}}  \label{eq:estimate-sigma-term2-t1}
\end{align}
where we concluded similarly as in \eqref{eq:time-integral1}. 

In the regime $\tau \in [t_0,t]$, we use the estimate \eqref{eq:estimate-sigma-term2-anfang} that shows 
\begin{align}
\int_{t_0}^t d\tau & \bigg( \int dy  \sup_{x \in \mathbb{R}^3}   \bigg\vert \int dz \; e^{-i(t-\tau) \Delta_x}\widetilde{K}_{1,\tau} (x; y) \sigma_{(t;\tau)} (y;z) \bigg\vert^2 \bigg)^{1/2} \notag \\
&\leq \int_{t_0}^t d\tau \frac{C}{(1 + \vert \tau \vert)^{3/2}} \| \sigma_{(t;\tau)} \|_{\rm HS} \leq \frac{C}{(1+t_0)^{1/2}} \sup_{ \tau\in [t_0,t]} \| \sigma_{(t;\tau)} \|_{\rm HS} \; . \label{eq:estimate-sigma-term2-t2}
\end{align}
Summarizing, \eqref{eq:estimate-sigma-term2-t1} together with \eqref{eq:estimate-sigma-term2-t2} implies that the second term of the r.h.s. of \eqref{eq:estimate-sigma-eins} is bounded by 
\begin{align}
\int_{s}^t d\tau & \bigg( \int dy  \sup_{x \in \mathbb{R}^3}   \bigg\vert \int dz \; e^{-i(t-\tau) \Delta_x}\widetilde{K}_{1,\tau} (x; y) \sigma_{(t;\tau)} (y;z) \bigg\vert^2 \bigg)^{1/2}\notag \\
\leq& \frac{C_1}{( 1 + \vert t- s\vert)^{3/2}} + \frac{C_2}{(1+t_0)^{1/2}}  \sup_{ \tau \in [t_0,t]}\| \sigma_{(t;\tau} \|_{\rm HS}  \; . \label{eq:estimate-sigma-term2-final}
\end{align}
It remains to bound the third term of the r.h.s. of \eqref{eq:estimate-sigma-eins}. A similar estimate as \eqref{eq:estimate-sigma-term2-final} follows from the same arguments as for the second term, discussed before, using the bounds for $\widetilde{K}_{2,\tau}$, instead of the bounds for $\widetilde{K}_{1,\tau}$ and $\gamma_{(t;s)} = \delta + \eta_{(t;s)}$, from Lemma \ref{lemma:K} so that we end up with 
\begin{align}
 \int_{s}^t d\tau & \bigg( \int dy  \sup_{x \in \mathbb{R}^3}   \bigg\vert \int dz \; e^{-i(t-\tau) \Delta_x}\widetilde{K}_{2,\tau} (x; y) \sigma_{(t;\tau)} (y;z) \bigg\vert^2\bigg)^{1/2}\notag \\
\leq& \frac{C_1}{(1+ \vert t-s\vert)^{3/2}} + \frac{C_2}{(1+t_0)^{1/2}}  \sup_{ \tau \in [t_0,t]} ( 1 + \| \eta_{(t;\tau} \|_{\rm HS} ) \; . \label{eq:estimate-sigma-term2-final}
\end{align}

Hence, summing up we get from \eqref{eq:estimate-sigma-term1-final} and \eqref{eq:estimate-sigma-term2-final} with 
\begin{align}
M(t_0,t) := \sup_{\tau \in [t_0,t]}  \bigg( \| \sigma_{(t;\tau)} \|_{\rm HS}
+ \| \nabla_1 \sigma_{(t,\tau)} \|_{\rm HS} + \| \Delta_1 \sigma_{(t;\tau)} \|_{\rm HS}  + \| \eta_{(t;\tau)} \|_{\rm HS} + 1 \bigg) 
\end{align}
that 
\begin{align}
 \| \sigma_{(t;s)} \|_{L^\infty \times L^2}   \leq \frac{C_1}{ (1 + \vert t- s\vert)^{3/2}} + \frac{C_2}{(t_0 +1)^{1/2}}  M (t_0;t) 
\end{align}
for some constants $C_1, C_2 >0 $. 

{To prove \eqref{eq:bound-zwei}, we are left with proving an appropriate upper bound for $M(t_0,s)$. For this we consider all the summands of $M(t_0,s)$ separately. We recall from the proof of Lemma \ref{lemma:gamma,sigma}, more precisely from \eqref{eq:gamma-HS-bound} that 
\begin{align}
\label{eq:gamma-bound-HS2}
\| \sigma_{(t;s)} \|_{\rm HS}^2  \leq C \int_s^t d\tau \;  \; \| K_\tau \|_{\rm HS} \| \sigma_{(t;\tau)} \|_{\rm HS}
\end{align}
leading with Lemma \ref{lemma:K} , to 
\begin{align}
\| \sigma_{(t;s)} \|_{\rm HS}^2  \leq C \int_s^t d\tau \; \frac{1}{( \tau + 1)^{3/2}}  \; \| \sigma_{(t;\tau)} \|_{\rm HS}  \; . 
\end{align} 
Since 
\begin{align}
\| \sigma_{(t;s)} \|_{\rm HS}^2  \leq \frac{C}{(t_0 +1)^{1/2}} \bigg( \sup_{\tau \in [t_0,t]} \; \| \sigma_{(t;\tau)} \|_{\rm HS}^2 + 1 \bigg)  
\end{align} 
we arrive at 
\begin{align}
\label{eq:gamma-bound-HS3}
\| \sigma_{(t;s)} \|_{\rm HS}  \leq \frac{C}{(t_0 +1)^{1/4}} \bigg( \sup_{\tau \in [t_0,t]} \; \| \sigma_{(t;\tau)} \|_{\rm HS} + 1 \bigg)  
\end{align} 
For the remaining contributions of $M(t_0,t)$, namely $\| \eta_{(t_0,t)} \|_{\rm HS}, \| \nabla \sigma_{(t_0,t)} \|_{\rm HS}$ and $\|\Delta \sigma_{(t_0,t)} \|_{\rm HS}$, we have similar estimates as \eqref{eq:gamma-bound-HS2} in \eqref{eq:eta-bound-HS}, \eqref{eq:nablagamma-bound-HS} resp. \eqref{eq:deltagamma-bound-HS} and get the corresponding bounds to \eqref{eq:gamma-bound-HS3} also for these quantities. We thus arrive at 
\begin{align}
M(t_0,t) \leq \frac{C_3}{(t_0 +1)^{1/4}} M(t_0;t)   + 1  
\end{align}
leading to 
\begin{align}
 \| \sigma_{(t;s)} \|_{L^\infty \times L^2}  + \bigg( 1- \frac{C_4}{ (1+t_0)^{1/4}} \bigg) \frac{1}{ 1+ \vert t- s\vert^{3/2} }  M(t_0,t) \leq \frac{C_1}{(1+ \vert t-s\vert)^{3/2}}  \; . 
\end{align}
Since $t_0 = \max \lbrace s,t-1\rbrace > t - 1 > T -1$, we can choose $T>2$ depending sufficiently large such that  $1-C_4 (1+t_0)^{-1/4} \leq \varepsilon$ and thus the desired estimate follows. }
\end{proof}

\section{Many-body Bogoliubov dynamics}
\label{sec:mbbogo}

The goal of this section is the proof of Theorem \ref{thm:norm}. To this end we first recall some properties of the fluctuation dynamics \eqref{def:flucdyn} and the many-body Bogoliubov dynamics on the (truncated) Fock space in Section \ref{sec:flucdyn} resp. Section \ref{sec:bogo-mb}. Then we prove Theorem \ref{thm:norm} in Section \ref{sec:proof-norm}. 

\subsection{Fluctuation dynamics}
\label{sec:flucdyn}

We recall that by definition of the fluctuation dynamics $\mathcal{W}_N (t;s)$ in \eqref{def:flucdyn}, it is an element of the truncated Fock space 
\begin{align}
\mathcal{F}_{\perp \varphi_t}^{\leq N} := \bigoplus_{n=0}^N L^2_{\perp \varphi_t}( \mathbb{R}^{3})^{\otimes n} \; . 
\end{align}
The truncated Fock space is equipped with modified creation and annihilation operators defined for $f,g \in L^2( \mathbb{R}^3)$ by 
\begin{align}
\label{def:b}
b(f) := \sqrt{\frac{N - \mathcal{N}}{N}} a(f)  , \quad \text{and} \quad b^*(g) = a^*(g) \sqrt{\frac{N - \mathcal{N}}{N}} 
\end{align}
that, in contrast to the standard creation and annihilation operators $a^*(f), a(g)$ defined on the Full bosonic Fock space $\mathcal{F} = \bigoplus_{n=0}^N L^2( \mathbb{R}^{3})^{\otimes n} $ and satisfying the standard canonical commutation relations 
\begin{align}
\label{eq:CCR-a}
\big[  a(f), a^*(g)\big] = \langle f,g, \rangle  , \quad \big[ a^*(f), a^*(g) \big] =\big[ a(f), a(g) \big]  =  0, 
\end{align}
leave the truncated Fock space invariant. However, this additional property of $b^*(f), b(g)$ comes with the price of a modified canonical commutation relation 
\begin{align}
\label{eq:CCR-b}
\big[  b(f), b^*(g)\big] = \langle f,g \rangle \bigg( 1- \frac{\mathcal{N}}{N} \bigg) - \frac{a^*(g) a(f)}{N}
\end{align}
that have an additional correction of $O(N)$ in the large particle limit $N \rightarrow \infty$. 

A straight forward calculation shows that the fluctuation dynamics $\mathcal{W}_N (t;s)$ satisfies 
\begin{align}
i\partial_t \mathcal{W}_N (t;s) = \mathcal{L}_N (t) \mathcal{W}_N (t;s)
\end{align}
with generator $\mathcal{L}_{N,t}$ given by the sum 
\begin{align}
\label{def:L}
\mathcal{L}_N (t) = \widetilde{\mathbb{H}}_t + \mathcal{R}_{N,t} 
\end{align}
of the leading order, quadratic (in modified creation and annihilation operators) $\widetilde{\mathbb{H}}_t$ 
\begin{align}
\label{def:Htilde}
\widetilde{\mathbb{H}}_t = d\Gamma (h_H (t) + \widetilde{K}_{1,t} ) + \frac{1}{2} \int dxdy \; \big[ \widetilde{K}_{2,t} (x;y) b_x^*b_y^* + {\rm h.c.} \big] 
\end{align}
where $\widetilde{K}_{1,t}, \widetilde{K}_{2,t}$ are given by \eqref{def:K} and the remainder $\mathcal{R}_{N,t}$ is a sum of four terms 
\begin{align}
\label{def:R}
\mathcal{R}_{N,t} = \sum_{j=1}^4 \mathcal{R}_{N,t}^{(j)}
\end{align}
that all come with a pre-factor of (at least) order $N^{-1/2}$ in the large particle limit and, more precisely given by, 
\begin{align}
\mathcal{R}_{N,t}^{(1)} = \frac{1-\mathcal{N}_+ (t)}{2 N } d\Gamma \big( q_t \big[ v* \vert \varphi_t \vert^2 + \widetilde{K}_{1,t} - \mu_t \big] q_t \big)  
\end{align}
with $\mu_t = 2 \langle \varphi_t, v * \vert \varphi_t \vert^2 \varphi_t \rangle $ and 
\begin{align}
\mathcal{R}_{N,t}^{(2)} =& \frac{ \mathcal{N}_+(t)}{\sqrt{N}} b(q_t \big[ ( v * \vert \varphi_t \vert^2 ) \varphi_t ) + {\rm h.c.}  \notag \\
\mathcal{R}_{N,t}^{(3)} =& \frac{1}{\sqrt{N}} \int dxdy \; v (x-y ) \varphi_t (x) a^*(q_{t,y}) a ({q_{t,x} } ) b(q_{t,y}) + {\rm h.c.} \notag \\
\mathcal{R}_{N,t}^{(4)} =& \frac{1}{2N} \int dxdy \; v (x-y )  a^*(q_{t,y}) a^*(q_{t,x}) a(q_{t,y}) a ({q_{t,x} })  \; . 
\end{align} 
In the following lemma we start with prove that moments of the number of particles are approximately presvered along the fluctuation dynamics $\mathcal{W}_N (t;s)$.

\begin{lemma}
\label{lemma:flucdyn}
Let $v \in L^1 ( \mathbb{R}^3) \cap L^2 ( \mathbb{R}^3)$, $t \in \mathbb{R}$ and $\varphi_t$ denote the solution to the Hartree equation with initial data $\varphi_0 \in H^1( \mathbb{R}^3)$ and $\| \varphi_0 \|_{L^2} = 1$. Then, for fixed $k \in \mathbb{N}$, there exists a constant $C_k>0$ (depending on $\| v \|_{L^1}$ and $\| v \|_{L^2}$ and $k \in \mathbb{N}$ only)  such that for any $\psi \in \mathcal{F}_{\perp \varphi_s}^{\leq N}$ 
\begin{align}
\label{eq:N-flucdyn-allg}
\langle \mathcal{W}_N (t;s) \psi, \; (\mathcal{N}+1)^k \mathcal{W}_N (t;s) \psi \rangle \leq C_k \int_s^t d\tau \| \varphi_\tau \|_{L^\infty} e^{\int_s^t  \| \varphi_\tau \|_{L^\infty} d \tau }  \langle \psi, ( \mathcal{N} + 1)^k \psi \rangle \; . 
\end{align}
for a constant $C_k >0$ independent of time and $N$. Moreover, let $v \in C_0^2 ( \mathbb{R}^3)$ and $\varphi_0 \in W^{1,\ell}( \mathbb{R}^3)$ for sufficiently large $\ell >0$. Then, for fixed $k \in \mathbb{N}$, there exists a constant $C_k>0$ (depending on $\| v \|_{L^1}$ and $\| v \|_{L^2}$ and $k \in \mathbb{N}$ only)  such that 
\begin{align}
\label{eq:N-flucdyn-disp}
\langle \mathcal{W}_N (t;s)\psi, \;  (\mathcal{N}+1)^k \mathcal{W}_N (t;s) \psi \rangle \leq C_k \bigg( \frac{1}{(1+ \vert s \vert)^{1/2}} - \frac{1}{(1+ \vert t \vert)^{1/2}} \bigg) \langle \psi, \; ( \mathcal{N} +1)^k \psi \rangle \; .  
\end{align}
\end{lemma}

\begin{proof}[Proof of Lemma \ref{lemma:flucdyn}] 
We start with the proof \eqref{eq:N-flucdyn-allg} for $k=1$. Then \eqref{eq:N-flucdyn-disp} $for k=1$ follows immediately with Proposition \ref{prop:dispersive}. 

The proof of \eqref{eq:N-flucdyn-allg} for $k=1$, follows standard techniques and is based on a Gronwall argument. To this end, we compute 
\begin{align}
i \partial_t \mathcal{W}_N^* (t;s) \big( \mathcal{N} + 1 \big) \mathcal{W}_N (t;s) =- \mathcal{W}_N^* (t;s) \big[ \mathcal{L}_{N,t}, \mathcal{N} \big] \mathcal{W}_N (t;s) 
\end{align}
and, thus, the goal is to control the commutator 
\begin{align}
\big[ \mathcal{L}_{N,t}, \; \mathcal{N} \big]  = \big[ \widetilde{\mathbb{H}}_t, \;  \mathcal{N} \big] + \big[ \mathcal{R}_{N,t}, \;  \mathcal{N} \big] \;
\end{align}
and we compute both commutators separately. For the first, we find 
\begin{align}
\label{eq:comm-Htilde}
\big[ \widetilde{\mathbb{H}}_t, \; \mathcal{N} \big] =& \int dxdy \big[ \widetilde{K}_{2,t} (x;y) b_x^* b_y^* - {\rm h.c.} \big] 
\end{align}
that be bounded for any $\psi, \xi \in \mathcal{F}_{\perp \varphi_s}^{\leq N}$ by 
\begin{align}
\big\vert \langle \xi, \big[ \widetilde{\mathbb{H}}_t, \; \mathcal{N} \big] \psi \rangle \big\vert \leq 4 \| \widetilde{K}_{2,t} \|_{\rm HS} \| ( \mathcal{N} + 1)^{1/2} \xi \| \; \|( \mathcal{N} + 1)^{1/2} \psi \|  \; . 
\end{align}
With Lemma \ref{lemma:K}, we arrive at 
\begin{align}
\label{eq:final-tildeH}
\big\vert \langle \xi, \big[ \widetilde{\mathbb{H}}_t, \; \mathcal{N} \big] \psi \rangle \big\vert \leq 4 \| v \|_{L^2} \| \varphi_t \|_{L^\infty}   \| ( \mathcal{N} + 1)^{1/2} \xi \| \; \|( \mathcal{N} + 1)^{1/2} \psi \| \; . 
\end{align}
We recall the splitting \eqref{def:R} of the remainder $\mathcal{R}_{N,t}$ to control its commutator with the number of particles operator and treat the single contributions separately. 

For the first contribution $\mathcal{R}_{N,t}^{(1)}$, the number of creation operators is equal to the number of annihilation operators, and therefore the commutator of the number of particles operator vanishes. 

For the second, we compute with the commutation relations 
\begin{align}
\label{eq:comm-R2}
\big[ \mathcal{R}_{N,t}^{(2)}, \; \mathcal{N} \big] =& - \frac{\mathcal{N}_+ (s)}{\sqrt{N}}  b( g_t ) - {\rm h.c.}
\end{align}
where we introduced the notation $g_t  = q_t \big[ (v* \vert \varphi_t \vert^2 ) \varphi_t \big] $. Since 
\begin{align}
\| g_t \|_{L^2} \leq \| v * \vert \varphi_t \vert^2 \|_{L^\infty} \| \varphi_t \|_{L^2} \leq \| v \|_{L^1}   \| \varphi_t \|_{L^\infty}^2 
\end{align}
we find using that $\mathcal{N}_+ (t) = \mathcal{N} \leq N $ on $\mathcal{F}^{\leq N}_{\perp \varphi_t}$ for any $\psi, \xi \in \mathcal{F}^{\leq N}_{\perp \varphi_t}$
\begin{align}
\vert \langle \psi, \big[ \mathcal{R}_{N,t}^{(2)}, \; \mathcal{N} \big] \xi \rangle \vert \leq& \frac{\| v \|_{L^1}}{\sqrt{N}} \| \varphi_t \|_{L^\infty}^2 \| ( \mathcal{N} + 1 ) \xi \| \; \|  ( \mathcal{N} + 1)^{1/2} \psi \| \notag \\
\leq&  2 \| v \|_{L^1} \| \varphi_t \|_{L^\infty}^2  \; \| ( \mathcal{N} + 1)^{1/2} \xi \| \; \|( \mathcal{N} + 1)^{1/2} \psi \| \, . \label{eq:final-R2} 
\end{align}
Furthermore, the commutation relations imply that 
\begin{align}
\label{eq:comm-R3}
\big[ \mathcal{R}_{N,t}^{(3)},  \; \mathcal{N} \big] =& - \frac{1}{\sqrt{N}} \int dxdy \; v(x-y) \varphi_t (x) a^*(q_{t,y}) a(q_{t,x}) b(q_{t,y} ) - {\rm h.c.} 
\end{align}
that we estimate for any $\psi, \xi \in \mathcal{F}_{\perp \varphi_t}^{\leq N}$ using Cauchy-Schwarz by 
\begin{align}
\big\vert \langle \xi, \big[ \mathcal{R}_{N,t}^{(3)},  \; \mathcal{N} \big] \psi \rangle \big\vert \leq& \frac{2}{\sqrt{N}} \bigg( \int dy \;  (v^2* \vert \varphi_t \vert^2) (y)  \; \| a(q_{t,y}) \xi \|^2  \bigg)^{1/2}  \bigg( \int dxdy  \;  \|a (q_{t,x}) b(q_{t,y}) \psi \|^2 \bigg)^{1/2} \notag \\ 
\leq& \frac{ 2 \| v \|_{L^2}}{\sqrt{N}} \| \varphi_t \|_{L^\infty} \| \mathcal{N}^{1/2} \xi \| \; \| (\mathcal{N} + 1) \psi \|  \notag \\
\leq& 2 \| v \|_{L^2} \| \varphi_t \|_{L^\infty} \| \mathcal{N}^{1/2} \xi \| \; \| (\mathcal{N} + 1)^{1/2} \psi \| \; .  \label{eq:final-R3}
\end{align}
With the observation $\big[ \mathcal{R}_{N,t}^{(4)}, \; \mathcal{N} \big] =0$, we thus arrive with \eqref{eq:final-tildeH}, \eqref{eq:final-R2} and \eqref{eq:R3-final} for any $\psi \in \mathcal{F}_{\perp \varphi_t}^{\leq N}$ at
\begin{align}
\label{eq:final-comm}
\big\vert \langle \xi, \; \big[ \mathcal{L}_{N,t}, \; \mathcal{N} \big] \psi \rangle \big\vert \leq \| \varphi_t \|_{L^\infty} \langle \xi, ( \mathcal{N} + 1) \psi \rangle 
\end{align}
for a constant $C>0$ independent of time and $N$. Then, the estimate \eqref{eq:N-flucdyn-allg} for $k=1$ follows from a Gronwall argument. 

We prove the general case $k \geq 1$ by induction. For this, we assume that \eqref{eq:N-flucdyn-allg} holds for fixed $k \in \mathbb{N}$. Then we compute 
\begin{align}
i \partial_t  &\mathcal{W}_N^* (t;s) \big(\mathcal{N}+1\big)^{k+1} \mathcal{W}_N (t;s)  \notag \\=& - \mathcal{W}_N^*(t;s) \big[ \mathcal{L}_{N,t}, \; (\mathcal{N}+1)^{k+1} \big] \mathcal{W}_N (t;s) \notag \\
=&\mathcal{W}_N^*(t;s)    (\mathcal{N}+1)^{(k+1)/2 }\big[ \mathcal{L}_{N,t}, \; (\mathcal{N}+1)^{(k+1)/2} \big] \mathcal{W}_N (t;s) - {\rm h.c.} \; .\label{eq:deriv-Nk}
\end{align}
To compute the commutator with $(\mathcal{N}+1)^{(k+1)/2}$ we use the identity
\begin{align}
\frac{1}{\sqrt{z}} = \frac{1}{\pi}\int_0^\infty d\kappa  \frac{\sqrt{\kappa}}{\kappa + z}
\end{align}
that allows to write 
\begin{align}
 (\mathcal{N}+1)^{(k+1)/2 } & \big[ \mathcal{L}_{N,t}, \; (\mathcal{N}+1)^{(k+1)/2} \big] \notag \\
& = \frac{1}{\pi} \int_0^\infty d\kappa \; \sqrt{\kappa} \frac{(\mathcal{N}+1)^{(k+1)/2 }}{\kappa + ( \mathcal{N} + 1)^{k+1}} \big[ \mathcal{L}_{N,t}, (\mathcal{N} + 1)^k \big] \frac{1}{\kappa + ( \mathcal{N} + 1)^{k+1}} \; . \label{eq:comm-k+1}
\end{align}
Since, by induction, we have 
\begin{align}
\big[ \mathcal{L}_{N,t}, (\mathcal{N} + 1)^k \big] = - \sum_{j=1}^{k+1} \binom{k+1}{j} {\rm ad}_{\mathcal{N}}^{(j)} \big( \mathcal{L}_{N,t} \big) \; \big( \mathcal{N} +1 )^{k+1-j}  \label{eq:id-sq}
\end{align}
where we introduced the notation ${\rm ad}_{\mathcal{N}}^{(j)}(A)$ for the $j$-the nested commmutator defined recursively for any operator $A$ through 
\begin{align}
{\rm ad}_{\mathcal{N}}^{(0)} (A)  = A  , \quad \text{and} \quad {\rm ad}_{\mathcal{N}}^{(j+1)} ( A) = \big[  {\rm ad}_{\mathcal{N}}^{(j)} ( A), \mathcal{N} \big] \; . 
\end{align}
We remark that \eqref{eq:final-comm} provides a bound of the first nested commutator ${\rm ad}_{\mathcal{N}}^{(1)} ( \mathcal{L}_{N,t})$ in sesqui-linear form. Since higher nested commutators ${\rm ad}_{\mathcal{N}}^{(j)} ( \mathcal{L}_{N,t})$ with $j \geq 1$ of the single contributions of $\mathcal{L}_{N,t}$ change the signs of \eqref{eq:comm-Htilde}, \eqref{eq:comm-R2}, \eqref{eq:comm-R3} only, the estimates \eqref{eq:final-tildeH}, \eqref{eq:final-R2} and \eqref{eq:final-R3} are valid for higher nested commutators ${\rm ad}_{\mathcal{N}}^{(j)} ( \mathcal{L}_{N,t})$, also, and thus arrive for any $\psi, \xi \in \mathcal{L}_{N,t}$ at the office. 
\begin{align}
\label{eq:adj}
\big\vert \langle \psi, {\rm ad}_{\mathcal{N}}^{(j)} ( \mathcal{L}_{N,t}) \xi \rangle \vert \leq  C_j \| \varphi_t \|_{L^\infty} \| ( \mathcal{N} + 1)^{1/2} \psi \| \; \| ( \mathcal{N} + 1)^{1/2} \xi \|  
\end{align}
where the constant $C_j >0$ is independent of time and $N$. We plug \eqref{eq:id-sq} back into \eqref{eq:comm-k+1} and find for any $\psi \in \mathcal{F}_{\perp \varphi_t}^{\leq N}$ with \eqref{eq:adj}  
\begin{align}
\big\vert \langle  \psi, \, (\mathcal{N}+1)^{(k+1)/2 } & \big[ \mathcal{L}_{N,t}, \; (\mathcal{N}+1)^{(k+1)/2} \big]  \psi \rangle \big\vert \notag \\
&\leq  C_k \| \varphi_t \|_{L^\infty} \sum_{j=1}^{k+1} \int_0^\infty d\kappa \; \sqrt{\kappa} \bigg\| \frac{( \mathcal{N} + 1)^{(k+2)/2}}{\kappa + ( \mathcal{N} + 1)^{k+1}} \psi \bigg\|  \; \bigg\| \frac{( \mathcal{N} + 1)^{k+3/2-j}}{\kappa + (\mathcal{N} + 1)^{k+1}} \psi \bigg\|  \notag \\
&\leq C_k \| \varphi_t \|_{L^\infty} \int_0^\infty d\kappa \; \sqrt{\kappa} \bigg\| \frac{( \mathcal{N} + 1)^{(k+2)/2}}{\kappa + ( \mathcal{N} + 1)^{k+1}} \psi \bigg\|  \; \bigg\| \frac{( \mathcal{N} + 1)^{k+1/2}}{\kappa + (\mathcal{N} + 1)^{k+1}} \psi \bigg\| \notag \\
&\leq C_k \| \varphi_t \|_{L^\infty} \int_0^\infty d\kappa \frac{\sqrt{\kappa}}{(1+\kappa)^{2 \frac{2k+1}{2k+2}}} \|  ( \mathcal{N} + 1)^{(k+1)/2} \psi \| \;  \|  ( \mathcal{N} + 1)^{k/2} \psi \|  \; . 
\end{align}
Since the integral over $\kappa$ converges for all $k \geq 1$, we find for any $\xi \in \mathcal{F}_{\perp \varphi_s}^{\leq N}$ from \eqref{eq:deriv-Nk} with the induction hypothesis 
\begin{align}
\big\vert i \partial_t \langle &  \xi, \; \mathcal{W}_N^*(t;s)  (\mathcal{N}+1)^{k+1} \mathcal{W}_N (t;s) \xi \rangle \big \vert \notag \\
=& \big\vert \langle  \xi, \; \mathcal{W}_N^*(t;s) \big[ \mathcal{L}_{N,t}, \; (\mathcal{N}+1)^{k+1} \big] \mathcal{W}_N (t;s) \xi \rangle \big \vert \notag \\
\leq& C_k \| \varphi_t \|_{L^\infty} \| ( \mathcal{N} + 1)^{k/2} \mathcal{W}_N(t;s) \psi \| \; \| ( \mathcal{N} + 1)^{(k+1)/2} \mathcal{W}_N(t;s) \psi \| \notag \\
\leq& C_k \| ( \mathcal{N} + 1)^{(k+1)/2} \mathcal{W}_N(t;s) \psi \| 
\end{align}
and \eqref{eq:N-flucdyn-allg} follows from a Gronwall argument. 
 \end{proof}

\subsection{Asymptotic many-body Bogoliubov dynamics}
\label{sec:bogo-mb}

In this section, we study properties of the asymptotic many-body Bogoliubov dynamics \eqref{def:bogo-asymp}. We recall that it is an element on the bosonic Fock space $\mathcal{F}$ and satisfies 
\begin{align}
i \partial_t \mathcal{W}_2 (t;s) = \mathcal{L}_{2,t} \mathcal{W}_2 (t;s) \label{def:bogo-mb}
\end{align}
with generator $\mathbb{H}_t$ quadratic in creation and annihilation operators 
\begin{align}
\label{def:L2}
\mathbb{H}_t =  d \Gamma ( -\Delta ) +  \int H_t(x;y) a_x^*a_y dxdy + \frac{1}{2}\int \big[ K_t (x;y) a^*_xa_y^* + \overline{K}_t (x;y) a_xa_y \big] dxdy \; 
\end{align}
and $H_t, K_t$ are given by \eqref{def:Hs,Ks}. 

A many-body Bogoliubov dynamics is characterized through its explicit action on creation and annihilation operators \eqref{eq:action-bogo}. This allows us to compute the action of the many-body Bogoliubov dynamics on powers of the number of particle operators. As a consequence, the evolution of powers of the number of particles operator along $\mathcal{W}_{2} (t;s)$ can be controlled. 

\begin{lemma}
\label{lemma:bogo}
Let $v \in L^1 ( \mathbb{R}^3) \cap L^2 ( \mathbb{R}^3)$, $t \in \mathbb{R}$ and $\varphi_t$ denote the solution to the Hartree equation with initial data $\varphi_0 \in H^1( \mathbb{R}^3)$ and $\| \varphi_0 \|_{L^2} = 1$. Then, for fixed $k \in \mathbb{N}$, there exists a constant $C_k>0$ (depending on $\| v \|_{L^1}$ and $\| v \|_{L^2}$ and $k \in \mathbb{N}$ only)  such that for any $\psi \in \mathcal{F}_{\perp \varphi_s}^{\leq N}$ 
\begin{align}
\langle \mathcal{W}_2 (t;s) \psi, \; (\mathcal{N}+1)^k \mathcal{W}_2 (t;s) \psi \rangle \leq C_k \int_s^t d\tau \| \varphi_\tau \|_{L^\infty} e^{\int_s^t  \| \varphi_\tau \|_{L^\infty} d \tau }  \langle \psi, ( \mathcal{N} + 1)^k \psi \rangle \; . 
\end{align}
for a constant $C_k >0$ independent of time and $N$. Moreover, let $v \in C_0^2 ( \mathbb{R}^3)$ and $\varphi_0 \in W^{1,\ell}( \mathbb{R}^3)$ for sufficiently large $\ell >0$. Then, for fixed $k \in \mathbb{N}$, there exists a constant $C_k>0$ (depending on $\| v \|_{L^1}$ and $\| v \|_{L^2}$ and $k \in \mathbb{N}$ only)  such that 
\begin{align}
\langle \mathcal{W}_2 (t;s)\psi, \;  (\mathcal{N}+1)^k \mathcal{W}_2 (t;s) \leq C_k \bigg( \frac{1}{(1+ \vert s \vert)^{1/2}} - \frac{1}{(1+ \vert t \vert)^{1/2}} \bigg) \langle \psi, \; ( \mathcal{N} +1)^k \psi \rangle \; .  
\end{align}
\end{lemma}

The Lemma follows from the proof of Lemma \ref{lemma:flucdyn} with the observation that any nested commutator of $\mathbb{H}_t$ with the number of particles operator $\mathcal{N}$ is of the form 
\begin{align}
{\rm ad}_{\mathcal{N}} \big( \mathbb{H}_t \big) = \int dxdy \; \big[ \widetilde{K}_{2,t} (x,y) a_x^*a_y^* \pm {\rm h.c.} \big]
\end{align}
and thus satisfies the bound \eqref{eq:adj}, so that the proof's argument applies also here. The second bound then immediately follows from the first together with Lemma \ref{lemma:K} and Proposition \ref{prop:dispersive}.


\subsection{Proof of Theorem \ref{thm:norm} }
\label{sec:proof-norm}

In this section, we prove Theorem \ref{thm:norm} on the validity of the Bogoliubov approximation for all times $t = o (N)$. In fact, the core of the proof of Theorem \ref{thm:norm} is to show that in the large particle limit, the modified creation and annihilation operators in the definition of the quadratic Hamiltonian $\widetilde{\mathbb{H}}_t$ can be replaced by standard ones (i.e. $\widetilde{\mathbb{H}}_t \approx \mathbb{H}_t $) and that the remainder $\mathcal{R}_{N,t}$ is negligible in the large particle limit. For this, we use the dispersive estimates on the symplectic Bogoliubov dynamics in Theorem \ref{thm:bogo-disp}. 

\begin{proof}[Proof of Theorem \ref{thm:norm}] 
First we note that while $\mathcal{L}_N (t)$ is defined on the truncated Fock space, the Bogoliubov dynamics is defined on the full bosonic Fock space. For this reason, we split 
\begin{align}
\| \mathcal{W}_N (t;0) \Omega &  -\mathcal{W}_2 (t;0) \Omega \|^2 \notag \\
& = \| \mathds{1}_{\mathcal{N} \leq N} \big( \mathcal{W}_{N} (t;0)  - \mathcal{W}_{2} (t;0) \big)\Omega   \|^2 +  \| \mathds{1}_{\mathcal{N}> N} \big( \mathcal{W}_N (t;0 ) - \mathcal{W}_{2}(t;0) \big) \Omega  \|^2 \; . \label{eq:norm-diff1}
\end{align}
To estimate the contribution in the large particle sector (i.e. the second term of the r.h.s. of \eqref{eq:norm-diff1}) , we split 
\begin{align}
\| \mathds{1}_{\mathcal{N}> N} \big( \mathcal{W}_N (t;0 ) - \mathcal{W}_{2}(t;0) \big) \Omega  \|^2 \leq \| \mathds{1}_{\mathcal{N} > N} \mathcal{W}_N (t;0) \Omega  \| + \| \mathds{1}_{\mathcal{N} > N} \mathcal{W}_2 (t;0) \Omega \| \label{eq:norm-diff-largeN}
\end{align}
and observe on the truncated Fock space $\mathcal{F}_{\perp \varphi_t}^{\leq N}$, we have the identity $\mathcal{N} = \mathcal{N}_+ (t)$ leading to 
\begin{align}
 \mathds{1}_{\mathcal{N} > N} \mathcal{W}_N (t;0) \Omega  = \mathds{1}_{\mathcal{N}_+ (t) > N} \mathcal{W}_N (t;0) \Omega =0 
\end{align}
and thus the first term of the r.h.s. of \eqref{eq:norm-diff-largeN} vanishes, and we are left with estimating the first. For this we use the inequality $\mathds{1}_{\mathcal{N} \leq N} \leq \frac{\mathcal{N}^2}{N^2}$ valid on the large particle sector 
\begin{align}
\| \mathds{1}_{\mathcal{N} > N} \mathcal{W}_2 (t;0) \Omega \|  \leq \frac{1}{N^2}  \| \mathcal{N}^2  \mathcal{W}_2 (t;0) \Omega  \| \; . 
\end{align}
With Lemma \ref{lemma:bogo}, we thus conclude 
\begin{align}
\| \mathds{1}_{\mathcal{N} > N} \mathcal{W}_2 (t;0) \Omega \|  \leq \frac{C}{N^2} 
\end{align}
for all $t \in \mathbb{R}$, finally yielding for the contribution in the large particle sector (i.e. the second term of the r.h.s. of \eqref{eq:norm-diff1} 
\begin{align}
\| \mathds{1}_{\mathcal{N}> N} \big( \mathcal{W}_N (t;0 ) - \mathcal{W}_{2}(t;0) \big) \Omega  \|^2  \leq \frac{C}{N^2} \; . \label{eq:large-N-final}
\end{align}

It remains to bound the first term of the r.h.s. of \eqref{eq:norm-diff1} that is to bound the difference of the fluctuation dynamics \eqref{def:flucdyn} and the Bogoliubov \eqref{def:bogo-asymp} restricted to the truncated Fock space. For this, we write 
\begin{align}
\| \big( \mathds{1}_{\mathcal{N} \leq N} \big( \mathcal{W}_N (t;0)\Omega - \mathcal{W}_2 (t;0)  \Omega \big) \|^2 = 2 \int_0^t ds  \frac{d}{ds}\; \Re \langle  \mathcal{W}_N (s;0)\Omega, \mathds{1}_{\mathcal{N}\leq N}\mathcal{W}_2 (s;0) \Omega \rangle 
\end{align}
and find with \eqref{def:flucdyn} and \eqref{def:bogo-asymp}
\begin{align}
\| \mathcal{W}_N &  (t;0)\Omega - \mathcal{W}_2 (t;0) \Omega \|^2 
\notag \\
=& 2 \int_0^t ds \; \Im \langle \mathcal{W}_N (s;0) \Omega, \big( \mathcal{L}_N (s) \mathds{1}_{\mathcal{N} \leq N} - \mathds{1}_{\mathcal{N} \leq N} \widetilde{\mathbb{H}} \big)  \mathcal{W}_2 (s;0) \Omega \rangle \; . 
\end{align}
and recalling the splitting \eqref{def:L}, we write 
\begin{align}
\|\mathds{1}_{\mathcal{N} \leq N} & \big (\mathcal{W}_N   (t;0)\Omega - \mathcal{W}_2 (t;0) \Omega  \big) \|^2 
\notag \\
=& 2 \int_0^t ds \; \Im \langle \mathcal{W}_N (s;0) \Omega,  \big(   \mathbb{H} - \widetilde{\mathbb{H}} \big) \mathds{1}_{\mathcal{N} \leq N} \mathcal{W}_2 (s;0) \Omega \rangle \notag \\
&+  2 \int_0^t ds \; \Im \langle \mathcal{W}_N (s;0) \Omega, \big[  \mathds{1}_{\mathcal{N} \leq N}, \widetilde{\mathbb{H}} \big]  \mathcal{W}_2 (s;0) \Omega \rangle \notag \\
&+ 2 \int_0^t ds \; \Im \langle \mathcal{W}_N (s;0)\Omega, \mathcal{R}_{N,s} \mathds{1}_{\mathcal{N} \leq N}  \mathcal{W}_2 (s;0) \Omega \rangle \; . \label{eq:norm-diff-21}
\end{align}
We estimate the three terms of the r.h.s. of \eqref{eq:norm-diff-21} separately. We start with the first and find with \eqref{def:flucdyn} resp. \eqref{def:bogo-asymp} 
\begin{align}
 \Im \langle  & \mathcal{W}_N (s;0)\Omega, \big( \mathbb{H} - \widetilde{\mathbb{H}} \big) \mathds{1}_{\mathcal{N} \leq N} \mathcal{W}_2 (s;0)\Omega \rangle \notag \\
 =& 4  \Im \langle \mathcal{W}_N (s;0) \Omega, \; \int dxdy K_{2,s} (x;y) ( b_x^*b_y^*   - a_x^*a_y^* ) \mathds{1}_{\mathcal{N} \leq N} \mathcal{W}_2 (s;0) \Omega \rangle \notag \\
 =& 4 \Im \langle \mathcal{W}_N (s;0) \Omega, \int dxdy K_{2,s} (x;y) a_x^*a_y^* \bigg( \frac{\sqrt{N - \mathcal{N}-1} \sqrt{N - \mathcal{N}}}{N}  -1 \bigg) \mathds{1}_{\mathcal{N} \leq N}\mathcal{W}_2 (s;0) \Omega \rangle 
\end{align}
that we estimate by Cauchy Schwarz 
\begin{align}
\big \vert \Im \langle  & \mathcal{W}_N (s;0) \Omega, \big( \mathbb{H} - \widetilde{\mathbb{H}} \big) \mathds{1}_{\mathcal{N} \leq N} \mathcal{W}_2 (s;0)\Omega \rangle \big \vert \notag \\
\leq& \int dy \bigg( \int dx  \; \vert K_{2,s} (x;y) \vert^2 \bigg)^{1/2} \bigg( \int dx  \; \| a_x \mathcal{W}_N (s;0) \Omega \|^2 \bigg)^{1/2} \notag \\
& \hspace{2cm} \times \big\| a_y^* \bigg( \frac{\sqrt{N - \mathcal{N}-1} \sqrt{N - \mathcal{N}}}{N}  -1 \bigg) \mathds{1}_{\mathcal{N} \leq N}\mathcal{W}_2 (s;0) \Omega \big\| \notag \\
\leq& \frac{\| K_{2,s} \|_{\rm HS}}{N} \| \mathcal{N}^{1/2} \mathcal{W}_N (s;0) \Omega \| \; \| (\mathcal{N} + 1)^{1/2} \mathcal{N} \mathds{1}_{\mathcal{N}\leq N} \mathcal{W}_2 (s;0) \Omega \| \; . 
\end{align}
With Lemma \ref{lemma:K} and Lemma \ref{lemma:bogo}, \ref{lemma:flucdyn} we then get 
\begin{align}
\big \vert \Im \langle  & \mathcal{W}_N (s;0) \Omega, \big( \mathbb{H} - \widetilde{\mathbb{H}} \big) \mathds{1}_{\mathcal{N} \leq N} \mathcal{W}_2 (s;0) \Omega \rangle \big \vert \leq \frac{C}{N} \frac{1}{(1 + \vert s \vert)^{3/2}}
\end{align}
and, integrating w.r.t. time, we find for the first term of the r.h.s. of \eqref{eq:norm-diff-21} 
\begin{align}
\int_0^t ds \; \big \vert \Im \langle  & \mathcal{W}_N (s;0) \Omega, \big( \mathbb{H} - \widetilde{\mathbb{H}} \big) \mathds{1}_{\mathcal{N} \leq N} \mathcal{W}_2 (s;0) \Omega \rangle \big \vert  \leq \frac{C}{N}  \label{eq:H-Hfinal}
\end{align}
for all $t \in \mathbb{R}$. 

For the second term of the r.h.s. of \eqref{eq:norm-diff-21} we find, with 
\begin{align}
\big[ \mathds{1}_{\mathcal{N} \leq N}, \widetilde{\mathbb{H}} \big] = \int dxdy \widetilde{K}_{2,s} (x;y) a_x^*a_y^* \big( \mathds{1}_{\mathcal{N} \leq N } - \mathds{1}_{\mathcal{N}  \leq N -2 } \big) = \int dxdy \widetilde{K}_{2,s} (x;y) a_x^*a_y^* \mathds{1}_{ N-1 \leq \mathcal{N} \leq N } 
\end{align}
and $\mathds{1}_{ N-1 \leq \mathcal{N} \leq N }  \leq \frac{\mathcal{N}}{(N-1)}$, with similar arguments as before, that 
\begin{align}
\big\vert \Im \langle & \mathcal{W}_N (s;0) \Omega,  \big[ \mathds{1}_{\mathcal{N} \leq N}, \widetilde{\mathbb{H}} \big] \mathcal{W}_2 (s;0) \Omega \rangle \big \vert \notag \\
\leq& \frac{\| K_{2,s} \|_{\rm HS}}{N} \| \mathcal{N}^{1/2} \mathcal{W}_N (s;0) \Omega \| \; \| (\mathcal{N} + 1)^{1/2} \mathcal{N} \mathcal{W}_2 (s;0) \Omega \| \; . 
\end{align}
leading since
\begin{align}
\label{eq:squareroot}
 \| (\mathcal{N} + 1)^{1/2} \mathcal{N} \mathcal{W}_2 (s;0) \Omega \| \leq  \| \mathcal{N}^{3/2} \mathcal{W}_2 (s;0) \Omega \| +   \| \mathcal{N} \mathcal{W}_2 (s;0) \Omega \| 
\end{align}
with Lemma \ref{lemma:flucdyn}, Lemma \ref{lemma:bogo} and Lemma \ref{lemma:K} for the second term of the r.h.s. of \eqref{eq:norm-diff1} to
\begin{align}
\int_0^t ds  \;  \big\vert \Im \langle & \mathcal{W}_N (s;0) \Omega,  \big[ \mathds{1}_{\mathcal{N} \leq N}, \widetilde{\mathbb{H}} \big] \mathcal{W}_2 (s;0) \Omega \rangle \big \vert  \leq \frac{C}{N}\int_0^t ds  \frac{1}{(1+ \vert s \vert)^{3/2}} \leq \frac{C}{N}  \label{eq:H-comm-final}
\end{align}
for all $t \in \mathbb{R}$. 

Next, we bound the third term of the r.h.s. of \eqref{eq:norm-diff1}. For this, we recall the splitting of the remainder $\mathcal{R}_{N,s}$ in \eqref{def:R} and consider all the single contributions separately.

We start with the first, for which we introduce the notation $A_s = q_s \big[ v* \vert \varphi_s \vert^2 + \widetilde{K}_{1,s} - \mu_s \big] q_s$ and estimate using $\mathcal{N}_+ (s) = \mathcal{N}$ on $\mathcal{F}_{\perp \varphi_s}^{\leq N}$
\begin{align}
\vert \Im \langle & \mathcal{W}_N (s;0) \Omega, \mathcal{R}_{N,s}^{(1)}\mathds{1}_{\mathcal{N} \leq N} \mathcal{W}_2 (s;0) \Omega \rangle \vert  \notag \\
&= \frac{1}{2N}\vert \langle  \mathcal{W}_N (s;0) \Omega, \big( 1- \mathcal{N}_+ (s)) d \Gamma (A_s ) \mathds{1}_{\mathcal{N} \leq N} \mathcal{W}_2 (s;0) \Omega \rangle \vert \notag \\
\leq& \frac{1}{2N} \| \mathcal{N}_+ (s) (\mathcal{N} +1 )^{-1/2} \mathcal{W}_N (s;0) \Omega \| \; \| d \Gamma (A_s ) (\mathcal{N}+1)^{1/2} \mathds{1}_{\mathcal{N} \leq N } \mathcal{W}_2 (s;0)\Omega \| \notag \\
&+ \frac{1}{2N} \| \mathcal{W}_N (s;0) \Omega \| \; \| d \Gamma(A_s ) \mathds{1}_{\mathcal{N} \leq N} \mathcal{W}_2 (s;0) \Omega \| \notag \\
\leq& \frac{1}{2N} \| A_s \|_{\rm op} \| \mathcal{N}^{1/2} \mathcal{W}_N (s;0) \Omega \| \; \| \mathcal{N} ( \mathcal{N} + 1)^{1/2} \mathcal{W}_2 (s;0) \Omega \| \notag \\
&+ \frac{1}{2N} \| A_s \|_{\rm op} \| \mathcal{W}_N (s;0) \Omega \| \; \| \mathcal{N} \mathcal{W}_2 (s;0) \Omega \| \; .
\end{align}
Since 
\begin{align}
\| A_s \|_{\rm op} \leq \| v* \vert \varphi_s \vert^2\|_{\rm op} + \| K_{1,s} \|_{\rm op} + \vert \mu_s \vert 
\end{align}
we arrive with Lemma \ref{lemma:K} at 
\begin{align}
\| A_s \|_{\rm op} \leq C \| v \|_{L^1} \| \varphi_s \|_{L^\infty}^2
\end{align}
leading with Lemma \ref{lemma:bogo}, Lemma \ref{lemma:flucdyn}, \eqref{eq:squareroot} to 
\begin{align}
\int_0^t ds \big\vert \Im \langle \mathcal{W}_N (t;s) \Omega, \; \mathcal{R}_{N,s}^{(1)} \mathds{1}_{\mathcal{N} \leq N} \mathcal{W}_2 (t;s) \Omega \rangle\vert  \leq  \frac{C}{N} \int_0^t ds \frac{1}{(1 + \vert s \vert)^3} \leq \frac{C}{N} \label{eq:R1-final}
\end{align}
for all $t \in \mathbb{R}$. 

The second contribution of the remainder is with the definition $g_s = q_s (v* \vert \varphi_s \vert^2 ) \varphi_s$ given by 
\begin{align}
\big\vert \Im & \langle \mathcal{W}_N (s;0) \Omega, \mathcal{R}_{N,s}^{(2)} \mathcal{W}_2 (s;0) \Omega \rangle \big\vert \notag \\
=& \frac{2}{\sqrt{N}} \big\vert \Im \langle \mathcal{W}_N (s;0) \Omega, \mathcal{N}_+ (s) b(g_s) \mathds{1}_{\mathcal{N} \leq N} \mathcal{W}_2 (s;0) \Omega \rangle \big\vert \notag \\ 
\leq& \frac{2}{\sqrt{N}} \| \mathcal{N}_+ (s) ( \mathcal{N} + 1)^{-1/2} \mathcal{W}_N (s;0) \Omega \| \; \| b(g_s) \mathcal{N}^{1/2} \mathds{1}_{\mathcal{N} \leq N} \mathcal{W}_2 (s;0) \Omega \| \notag \\
\leq& \frac{2}{\sqrt{N}}  \| g_s \|_{L^2} \| \mathcal{N}^{1/2} \mathcal{W}_N (s;0) \Omega \| \; \| \mathcal{N}  \mathds{1}_{\mathcal{N} \leq N} \mathcal{W}_2 (s;0) \Omega \| \; . 
\end{align}
Since 
\begin{align}
\| g_s \|_{L^2}^2 \leq \| v*\vert \varphi_s \vert^2 \|_{L^\infty} \| \varphi_s \|_{L^2} \leq \| \varphi_s \|_{L^\infty}^2 \| v \|_{L^1} 
\end{align}
we get with Lemma \ref{lemma:bogo}, Lemma \ref{lemma:flucdyn} and Proposition \ref{prop:dispersive} 
\begin{align}
\int_0^t ds \; \big\vert \Im  \langle \mathcal{W}_N (s;0)\Omega, \mathcal{R}_{N,s}^{(2)} \mathcal{W}_2 (s;0) \Omega \rangle \big\vert \leq& \frac{C}{\sqrt{N}} \int_0^t ds\; \frac{1}{( 1 +\vert s \vert)^{3}} \leq  \frac{C}{\sqrt{N}}  \label{eq:R2-final}
\end{align}
for all $t \in \mathbb{R}$.

For the third contribution, we write
\begin{align}
\Im \langle & \mathcal{W}_N(s;0) \Omega, \mathcal{R}_{N,s}^{(3)} \mathds{1}_{\mathcal{N} \leq N} \mathcal{W}_2 (s;0) \Omega \rangle \notag \\
=& \Im \langle \mathcal{W}_{N} (s;0) \Omega - \mathds{1}_{\mathcal{N} \leq N} \mathcal{W}_2 (s;0) \Omega , \mathcal{R}_N^{(3)} \mathds{1}_{\mathcal{N} \leq N} \mathcal{W}_2 (s;0) \Omega  \rangle \notag \\
=& \Im \mathds{1}_{\mathcal{N} \leq N} \big( \langle \mathcal{W}_{N} (s;0) \Omega -  \mathcal{W}_2 (s;0) \Omega \big), \mathcal{R}_N^{(3)} \mathds{1}_{\mathcal{N} \leq N} \mathcal{W}_2 (s;0) \Omega  \rangle \notag \\
&+  \Im  \langle \mathds{1}_{\mathcal{N} > N} \mathcal{W}_{N} (s;0) \Omega , \mathcal{R}_N^{(3)} \mathds{1}_{\mathcal{N} \leq N} \mathcal{W}_2 (s;0) \Omega  \rangle 
\end{align}
that we estimate by 
\begin{align}
\big\vert \Im \langle & \mathcal{W}_N(s;0) \Omega, \mathcal{R}_{N,s}^{(3)} \mathds{1}_{\mathcal{N} \leq N} \mathcal{W}_2 (s;0) \Omega \rangle \big\vert \notag \\
\leq& \bigg(  \| \mathds{1}_{\mathcal{N} \leq N} \big(  \mathcal{W}_{N} (s;0) \Omega -  \mathcal{W}_2 (s;0) \Omega \big) \| +\| \mathds{1}_{\mathcal{N} > N} \mathcal{W}_{N} (s;0) \Omega \| \bigg)  \; \| \mathcal{R}_N^{(3)} \mathds{1}_{\mathcal{N} \leq N} \mathcal{W}_2 (s;0) \Omega   \| \; . \label{eq:estimate-R3-1}
\end{align} 
First we observe that the first term $\| \mathds{1}_{\mathcal{N} \leq N} \big(  \mathcal{W}_{N} (s;0) \Omega -  \mathcal{W}_2 (s;0) \Omega \big) \| $ is exactly the quantity we are estimating in this part of the proof and we leave it as it is. The second term however can be easily, using $\mathds{1}_{\mathcal{N} \> N} \leq \frac{\mathcal{N}^k}{N^k} $ for any $k \in \mathbb{N}$ and Lemma \ref{lemma:flucdyn} by 
\begin{align}
\| \mathds{1}_{\mathcal{N} > N} \mathcal{W}_{N} (s;0) \Omega \| \leq \frac{1}{N}  \| \mathcal{N}  \mathcal{W}_{N} (s;0)  \| \leq  \frac{C}{N} \; \label{eq:estimate-R3-2}
\end{align} 
for all $s \in \mathbb{R} $and we are left with estimating the last term of the r.h.s. of \eqref{eq:estimate-R3-1}. For this we recall that 
\begin{align}
\mathcal{R}_{N,s}^{(3)} = \mathcal{R}_{N,s}^{(3,1)} + \mathcal{R}_{N,s}^{(3,2)} 
\end{align}
with 
\begin{align}
\mathcal{R}_{N,s}^{(3,1)} = \frac{1}{\sqrt{N}} \int dxdy \; v(x-y) \varphi_t (x) a^*(q_{t,y}) a(q_{t,x}) b(q_{t,y} )  
\end{align}
and $\mathcal{R}_{N,s}^{(3,2)} $ is the adjoint of $\mathcal{R}_{N,s}^{(3,1)} $. Thus we have 
\begin{align}
\|  \mathcal{R}_N^{(3)} \mathds{1}_{\mathcal{N} \leq N} \mathcal{W}_2 (s;0) \Omega   \| \leq& \| \mathcal{R}_N^{(3,1)} \mathds{1}_{\mathcal{N} \leq N} \mathcal{W}_2 (s;0) \Omega   \| + \| \mathcal{R}_N^{(3,2)} \mathds{1}_{\mathcal{N} \leq N} \mathcal{W}_2 (s;0) \Omega   \| \label{eq:estimate-R3-3}
\end{align}
and it remains to estimate both terms. We start with the first that is, using that $\mathcal{W}_2 (s;0)$ is an element of the truncated Fock space built over the orthogonal complement of $\varphi_s$ in $L^2$, where $a^*(q_{s,x})  =a_x $, given by 
\begin{align}
\|   \mathcal{R}_N^{(3,1)}  &  \mathds{1}_{\mathcal{N} \leq N}\mathcal{W}_2 (s;0) \Omega   \|^2 \notag \\
=& \int dxdy dzdw \; v(x-y) v(z-w) \varphi_s (x) \overline{\varphi}_s (z) \notag \\
& \hspace{2cm} \times \langle  \mathds{1}_{\mathcal{N} \leq N} \mathcal{W}_2 (s;0)  \Omega, b^*_za^*_wa_w a^*_y a_xb_y \mathds{1}_{\mathcal{N} \leq N} \mathcal{W}_2 (s;0)  \Omega \rangle  \notag \\
=& \int dxdy dzdw \; v(x-y) v(z-w) \varphi_s (x) \overline{\varphi}_s (z) \notag \\
& \hspace{2cm} \times \langle  \mathds{1}_{\mathcal{N} \leq N} \mathcal{W}_2 (s;0)  \Omega, b^*_za^*_w a^*_y a_w a_xb_y \mathds{1}_{\mathcal{N} \leq N} \mathcal{W}_2 (s;0)  \Omega \rangle  \notag \\
&+ \int dxdy dz \; v(x-y) v(z-y) \varphi_s (x) \overline{\varphi}_s (z) \langle  \mathds{1}_{\mathcal{N} \leq N} \mathcal{W}_2 (s;0)  \Omega, b^*_za^*_y  a_xb_y \mathds{1}_{\mathcal{N} \leq N} \mathcal{W}_2 (s;0)  \Omega \rangle \notag \\ 
\leq& \int dz dw dy  \;  ( v^2 * \vert \varphi_t \vert^2) (y) \; \| a_y a_wb_z \mathds{1}_{\mathcal{N} \leq N} \mathcal{W}_2 (s;0)  \Omega \|^2 \notag \\
 &+  \int dydz ( v* \vert \varphi_t\vert^2 )(y) \| a_y b_z \mathds{1}_{\mathcal{N} \leq N} \mathcal{W}_2 (s;0)  \Omega \|^2 \notag \\
\leq& \| v^2 * \vert \varphi_s \vert^2 \|_{L^\infty}  \| ( \mathcal{N} + 1)^{3/2} \mathcal{1}_{\mathcal{N} \leq N}  \mathcal{W}_2 (s;0)  \Omega \| \notag \\
 &= \| v^2 * \vert \varphi_s \vert^2 \|_{L^\infty}  \| ( \mathcal{N} + 1)^{3/2}  \mathcal{W}_2 (s;0)  \Omega \|
\end{align}
and since $\| v^2 * \vert \varphi_s \vert^2 \|_{L^\infty} \leq C \| v^2 \|_{L^1} \; \| \vert \varphi_s \vert^2 \|_{L^\infty} \leq C \| \varphi_s \|_{L^\infty}^2  \| v \|_{L^2} $, we find with Proposition \ref{prop:dispersive} and Lemma \ref{lemma:bogo} that 
\begin{align}
\|   \mathcal{R}_N^{(3,1)}   \mathds{1}_{\mathcal{N} \leq N}\mathcal{W}_2 (s;0) \Omega  \|^2 \leq \frac{C}{N} \frac{1}{(1 + \vert s \vert)^3} \; . 
\end{align}
The second term $\|   \mathcal{R}_N^{(3,2)}   \mathds{1}_{\mathcal{N} \leq N}\mathcal{W}_2 (s;0) \Omega  \|^2$ of \eqref{eq:estimate-R3-3} can be bounded similarly and we arrive with \eqref{eq:estimate-R3-2} from \eqref{eq:estimate-R3-3} at 
\begin{align}
\int_0^t ds \; \big\vert \Im \langle & \mathcal{W}_N(s;0) \Omega, \mathcal{R}_{N,s}^{(3)} \mathds{1}_{\mathcal{N} \leq N} \mathcal{W}_2 (s;0) \Omega \rangle \big\vert \notag \\
&\leq \frac{C}{\sqrt{N}} \int_0^t ds \; \frac{1}{(1+ \vert s \vert)^{3/2}}   \| \mathds{1}_{\mathcal{N} \leq N}  \mathcal{W}_{N} (s;0) \Omega \|^2 -  \mathcal{W}_2 (s;0) \Omega \| + \frac{C}{N^{3/2}} \label{eq:R3-final} \; . 
\end{align}

The forth term of the remainder contains as many creation as annihilation operators, thus we write 
\begin{align}
\Im \langle &  \mathcal{W}_N (s;0) \Omega, \mathcal{R}_{N,s}^{(4)} \mathds{1}_{\mathcal{N} \leq N} \mathcal{W}_2 (s;0) \Omega \rangle \notag \\
=& \Im \langle  \mathcal{W}_N (s;0) \Omega, \mathds{1}_{\mathcal{N} \leq N} \mathcal{R}_{N,s}^{(4)} \mathds{1}_{\mathcal{N} \leq N} \mathcal{W}_2 (s;0) \Omega \rangle \notag \\
=& \Im \langle  \mathds{1}_{\mathcal{N} \leq N} \big(\mathcal{W}_N (s;0) - \mathcal{W}_2 (s;0) \big) \Omega, \mathds{1}_{\mathcal{N} \leq N} \mathcal{R}_{N,s}^{(4)} \mathds{1}_{\mathcal{N} \leq N} \mathcal{W}_2 (s;0) \Omega \rangle 
 \end{align}
where we used in the last step that $\mathcal{R}_{N,s}^{(4)}$ is self-adjoint. We furthermore write  
\begin{align}
\Im \langle &  \mathcal{W}_N (s;0) \Omega, \mathcal{R}_{N,s}^{(4)} \mathds{1}_{\mathcal{N} \leq N} \mathcal{W}_2 (s;0) \Omega \rangle \notag \\
=& \frac{1}{2N}\int dxdy \; v(x-y) \; \notag \\
& \hspace{1cm}  \times \langle \mathds{1}_{\mathcal{N} \leq N} \big(\mathcal{W}_N (s;0) - \mathcal{W}_2 (s;0) \big) \Omega, a^*(q_{s,x})a^*(q_{s,y}) a(q_{s,x}) a(q_{s,y}) \mathds{1}_{\mathcal{N} \leq N} \mathcal{W}_2 (s;0) \Omega \rangle 
 \end{align}
that we estimate with 
 \begin{align}
 \big \vert \Im \langle &  \mathcal{W}_N (s;0) \Omega, \mathcal{R}_{N,s}^{(4)} \mathds{1}_{\mathcal{N} \leq N} \mathcal{W}_2 (s;0) \Omega \rangle \big\vert \notag \\
 &\leq \frac{1}{2N} \bigg( \int dx dy \; \vert v(x-y) \vert \;  \|  a(q_{s,x})a(q_{s,y})  \mathds{1}_{\mathcal{N} \leq N} \big(\mathcal{W}_N (s;0) - \mathcal{W}_2 (s;0) \big) \Omega  \|^2 \bigg)^{1/2} \notag \\
 & \hspace{1cm} \times \bigg( \int dx dy \; \vert v(x-y) \vert \;  \|  a(q_{s,y})a(q_{s,x}) \mathcal{W}_2 (s;0) \Omega \|^2 \bigg)^{1/2}   
 \end{align}
leading with standard arguments to 
\begin{align}
 \big \vert \Im \langle &  \mathcal{W}_N (s;0) \Omega, \mathcal{R}_{N,s}^{(4)} \mathds{1}_{\mathcal{N} \leq N} \mathcal{W}_2 (s;0) \Omega \rangle \big\vert \notag \\
 \leq& \frac{C}{N} \| v \|_{\infty}^{1/2} \| (\mathcal{N}  + 1)\big(\mathcal{W}_N (s;0) - \mathcal{W}_2 (s;0) \big) \Omega \| \;  \notag \\
 & \hspace{2cm} \times \bigg( \int dxdy \; \vert v(x-y) \vert \;  \| a(q_{s,y} ) a(q_{s,x}) \mathcal{W}_2 (s;0) \Omega \|^2 \bigg)^{1/2}\label{eq:R4-step1}
\end{align}
On the one hand, we get with Lemma \ref{lemma:bogo} and Lemma \ref{lemma:flucdyn} 
\begin{align}
\| ( \mathcal{N} + 1 ) \big( \mathcal{W}_N (s;0) - \mathcal{W}_2 (s;0) \big) \Omega \| \leq C \label{eq:R4-1-final} 
\end{align}
for all $s \in \mathbb{R}$. On the other hand, we find 
\begin{align}
\int dxdy &  \; \vert v(x-y) \vert \;  \| a(q_{s,y} ) a(q_{s,x}) \mathcal{W}_2 (s;0) \Omega \|^2 \notag \\
=& \int dxdy \; \vert v(x-y) \vert \; \langle \Omega, \mathcal{W}_2^* (s;0) a^*(q_{s;y} ) a^*(q_{s,x}) a(q_{s,y}) a(q_{s,x}) \mathcal{W}_2(s;0) \Omega \rangle 
\end{align}
The Fock space vector $\mathcal{W}_2(s;0) \Omega$ is, by definition, an element of the truncated Fock space $\mathcal{F}_{\perp \varphi_s}^{\leq N}$ built over the orthogonal complement of $\varphi_s$ in $L^2$, thus $a(q_{s,x}) \mathcal{W}_2(s;0) \Omega = a_x \mathcal{W}_2(s;0) \Omega$ and we are left with estimating 
\begin{align}
\int dxdy &  \; \vert v(x-y) \vert \;  \| a(q_{s,y} ) a(q_{s,x}) \mathcal{W}_2 (s;0) \Omega \|^2 \notag \\
=& \int dxdy \; \vert v(x-y) \vert \; \langle \Omega, \mathcal{W}_2^* (s;0) a^*_y a^*_x a_y a_x \mathcal{W}_2(s;0) \Omega \rangle \; . 
\end{align}
Infact, we will compute this expectation value explicitely using the properties of the Bogoliubov dynamics.  More precisely, writing in abuse of notation
\begin{align}
\sigma_{(s;0)} (x,y) = \sigma_y ( x) 
\end{align}
we get 
\begin{align}
\int dxdy &  \; \vert v(x-y) \vert \;  \| a(q_{s,y} ) a(q_{s,x}) \mathcal{W}_2 (s;0) \Omega \|^2 \notag \\ 
=& \int dxdy \; \vert v(x-y) \vert \langle \Omega, a( \sigma_x ) \big( a^*( \gamma_y) + a( \sigma_y) \big) \big( a ( \gamma_x) + a^*( \sigma_x) \big) a^* ( \sigma_y) \Omega \rangle  \; . 
\end{align}
The expectation value vanishes whenever the number of creation operators does not match the number of annihilation operators, therefore with the commutation relations 
\begin{align}
\int dxdy &  \; \vert v(x-y) \vert \;  \| a(q_{s,y} ) a(q_{s,x}) \mathcal{W}_2 (s;0) \Omega \|^2 \notag \\ 
=& \int dxdy \; \vert v(x-y) \vert \langle \Omega, a( \sigma_x )  a^*( \gamma_y)  a ( \gamma_x)  a^* ( \sigma_y) \Omega \rangle  \notag \\
&+  \int dxdy \; \vert v(x-y) \vert \langle \Omega, a( \sigma_x ) a( \sigma_y) a^*( \sigma_x)  a^* ( \sigma_y) \Omega \rangle  \notag \\
=&  \int dydx \vert v(x-y) \vert \; \bigg( \vert \langle \sigma_y, \sigma_x \rangle \vert^2 + \| \sigma_x \|_{L^2}^2 \| \sigma_y \|_{L^2}^2 + \langle \sigma_x, \gamma_y \rangle \langle \gamma_x, \sigma_y \rangle \bigg) \label{eq:R4-ausrechnen}
\end{align} 
Next we show that the three terms of the r.h.s. of \eqref{eq:R4-ausrechnen} are bounded in terms of $\| \sigma \|_{L^\infty \times L^2}$. In fact, the first two terms of the r.h.s. are bounded by
\begin{align}
\int dydx & \vert v(x-y) \vert \; \bigg( \vert \langle \sigma_y, \sigma_x \rangle \vert^2 + \| \sigma_x \|_{L^2}^2 \| \sigma_y \|_{L^2}^2  \bigg) \notag \\
\leq& 2 \int dxdy \; \vert v(x-y) \vert \| \sigma_x\|_{L^2}^2 \|\sigma_y \|_{L^2}^2 \notag \\
\leq& 2 \| \sigma \|_{L^\infty \times L^2}^2 \int dxdy \vert v(x-y) \vert \; \| \sigma \|_{L^2}^2 \notag \\
&\leq 2 \| \sigma \|_{L^\infty \times L^2}^2 \| v \|_{L^1} \| \sigma \|_{\rm HS}^2 \; . \label{eq:R4-1}
\end{align}
For the third term of the r.h.s. of \eqref{eq:R4-ausrechnen} we recall that we have the decomposition $\gamma (x;y) = \delta (x-y) + \eta (x;y)$ and $\eta \in L^2 \times L^2$ is bounded in Hilbert-Schmidt norm independent of time (see Lemma \ref{lemma:gamma,sigma}). With this decomposition the third term of the r.h.s. of \eqref{eq:R4-ausrechnen} is of the form 
\begin{align}
 \int dydx  & \vert v(x-y) \vert \;  \langle \sigma_x, \gamma_y \rangle \langle \gamma_x, \sigma_y \rangle \notag \\
=& \int dxdy \; \vert v(x-y) \vert \bigg( \vert \sigma(x;y) \vert^2 + 2 \Re \sigma (x;y) \langle \eta_x, \sigma_y \rangle + \langle \sigma_x, \eta_y \rangle \langle \eta_x, \sigma_y \rangle \bigg) \notag 
\end{align} 
and  bounded by 
\begin{align}
 \int dydx  & \vert v(x-y) \vert \;  \langle \sigma_x, \gamma_y \rangle \langle \gamma_x, \sigma_y \rangle \notag \\ 
\leq &  4  \int dxdy \; \vert v(x-y) \vert  \bigg( \vert \sigma(x;y) \vert^2 + \| \eta_x \|_{L^2}^2 \| \sigma_y \|_{L^2}^2  \bigg) \notag \\
\leq& 4 \| v \|_{\infty} \| \sigma \|_{L^\infty \times L^2}^2 \bigg( \| \sigma \|_{\rm HS}^2 + \| \eta \|_{\rm HS}^2 \bigg) \; . \label{eq:R4-2}
\end{align} 
Hence, \eqref{eq:R4-1} and \eqref{eq:R4-2} imply together with Lemma \ref{lemma:gamma,sigma} that 
\begin{align}
\int dxdy &  \; \vert v(x-y) \vert \;  \| a(q_{s,y} ) a(q_{s,x}) \mathcal{W}_2 (s;0) \Omega \|^2 \notag \\  
\leq& \frac{C}{ (1+ \vert s \vert)^3} \label{eq:R4-ausrechnen-final}
\end{align}
for some constant $C>0$. Summing up \eqref{eq:R4-1-final} and \eqref{eq:R4-ausrechnen-final} we get for the forth remainder (i.e. the term from \eqref{eq:R4-step1}) 
\begin{align}
\big\vert \Im \langle \mathcal{W}_N (s;0) \Omega, \mathcal{R}_{N,s}^{(4)} \mathds{1}_{\mathcal{N} \leq N} \mathcal{W}_2 (s;0) \Omega \rangle \big\vert \leq  \frac{C}{N} \frac{1}{(1+\vert s \vert)^{3/2}} 
\end{align}
and, integrating over time $s \in [0,t]$ we arrive at 
\begin{align}
\int_0^t ds\; \big\vert \Im \langle \mathcal{W}_N (s;0) \Omega, \mathcal{R}_{N,s}^{(4)} \mathds{1}_{\mathcal{N} \leq N} \mathcal{W}_2 (s;0) \Omega \rangle \big\vert \leq  \frac{C}{N}  \label{eq:R4-final}
\end{align}
for all times $t \in \mathbb{R}$.

Summing up \eqref{eq:R1-final}, \eqref{eq:R2-final}, \eqref{eq:R3-final} and \eqref{eq:R4-final}, we arrive at the final estimate of Theorem \ref{thm:norm}. 

\end{proof}

\subsection{Proof of Corollary \ref{cor:largetimes}} 

\label{sec:cor}

In this section we finally prove Corollary that is a consequence of Theorem \ref{thm:norm} on the Bogoliubov approximation of the fluctuation dynamics and the dispersive estimates of the symplectic Bogoliubov dynamics in Theorem \ref{thm:bogo-disp}. 

\begin{proof}[Proof of Corollary \ref{cor:largetimes}] 
In fact, Theorem \ref{thm:norm} together with the triangle inequality shows that 
\begin{align}
\| &  \mathcal{U}_{N,t} \psi_{N,t} -  e^{-i (t-t_0) d\Gamma ( \Delta)} \mathcal{W}_2 (t_0;0)\Omega  \|^2  \notag \\
&\leq \| \mathcal{U}_{N,t} \psi_{N,t} - \mathcal{W}_2 (t;0)  \Omega  \|^2 + \| \mathcal{W}_2 (t;0)  \Omega  - e^{-i (t-t_0) d\Gamma ( \Delta)} \mathcal{W}_2 (t_0;0) \Omega  \|^2 \notag \\
 &\leq  \| \mathcal{W}_2 (t;0)  \Omega  -  e^{-i (t-t_0) d\Gamma ( \Delta)} \mathcal{W}_2 (t_0;0)\Omega  \|^2 + \frac{C}{N}
\end{align}
for some $C>0$, and therefore for the proof of Corollary \ref{cor:largetimes} it suffices to compare the many-body Bogoliubov dynamics with the free dynamics. For this, we follow the strategy of Theorem's \ref{thm:norm} proof and write with Duhamel's formula 
\begin{align}
\| &  \mathcal{W}_2 (t;0) \Omega  - e^{-i (t-t_0) d\Gamma ( \Delta)} \mathcal{W}_2 (t_0;0)  \Omega \|^2 \notag \\
=& 2 \Im  \int_{t_0}^t \; ds \; \langle  e^{-i (s-t_0) d\Gamma ( \Delta)} \mathcal{W}_2 (t_0;0) \Omega,  d\Gamma ( H_s - \Delta )  \mathcal{W}_2 (s;0)  \Omega \rangle  \notag \\
&+  \Im  \int_{t_0}^t \; ds \; \langle  e^{-i (s-t_0) d\Gamma ( \Delta)} \mathcal{W}_2 (t_0;0) \Omega,  \int dxdy \big[  K_s (x;y) a_x^*a_y^* + {\rm h.c.}\big]   \mathcal{W}_2 (s;0)  \Omega \rangle . \label{eq:corr1}
\end{align}
Since $d\Gamma ( H_s - \Delta)$ is a self-adjoint operator, we can write the first line of the r.h.s. as 
\begin{align}
\Im   & \int_{t_0}^t \; ds \; \langle  e^{-i (s-t_0) d\Gamma ( \Delta)} \mathcal{W}_2 (t_0;0) \Omega,  d\Gamma ( H_s - \Delta )  \mathcal{W}_2 (s;0)  \Omega \rangle \notag \\
=& \Im  \int_{t_0}^t \; ds \; \langle  \big( e^{-i (s-t_0) d\Gamma ( \Delta)} \mathcal{W}_2 (t_0;0) \Omega - \mathcal{W}_2 (s;0)  \Omega  \big),  d\Gamma ( H_s - \Delta )  \mathcal{W}_2 (s;0)  \Omega \rangle 
\end{align}
and estimate the term with similar ideas as in the proof of Theorem \ref{thm:norm} by 
\begin{align}
 & \int_{t_0}^t \; ds \; \big\vert \langle  e^{-i (s-t_0) d\Gamma ( \Delta)} \mathcal{W}_2 (t_0;0) \Omega,  d\Gamma ( H_s - \Delta )  \mathcal{W}_2 (s;0)  \Omega \rangle \big\vert \notag \\
\leq &  \int_{t_0}^t ds \; \|   e^{-i (s-t_0) d\Gamma ( \Delta)} \mathcal{W}_2 (t_0;0) \Omega - \mathcal{W}_2 (s;0)  \Omega  \| \; \| H_s - \Delta \|_{\rm op} \; \| \mathcal{N} \mathcal{W}_2 (s;0)  \Omega \| \; . 
\end{align}
Since 
\begin{align}
\| H_s  - \Delta \|_{\rm op} \leq \| v*  \vert\varphi_s \vert^2\|_{\infty} +  \| \widetilde{K}_{1,s} \|_{\rm op} \leq \frac{C}{(1+ \vert s \vert)^3}
\end{align}
from Lemma \ref{lemma:K} and Proposition \ref{prop:dispersive}, we arrive, using that $\| \mathcal{N} \mathcal{W}_2 (s;0)  \Omega \|  \leq C $ for all $s \in \mathbb{R}$ by Lemma \ref{lemma:bogo}, at 
\ref{thm:norm} by 
\begin{align}
 & \int_{t_0}^t \; ds \; \big\vert \langle  e^{-i (s-t_0) d\Gamma ( \Delta)} \mathcal{W}_2 (t_0;0) \Omega,  d\Gamma ( H_s - \Delta )  \mathcal{W}_2 (s;0)  \Omega \rangle \big\vert \notag \\
\leq &  C \int_{t_0}^t ds \; \|   e^{-i (s-t_0) d\Gamma ( \Delta)} \mathcal{W}_2 (t_0;0) \Omega - \mathcal{W}_2 (s;0)  \Omega \| \;  \frac{1}{(1+ \vert s\vert)^3} \;  \label{eq:corr2}
\end{align}
that is the final estimate for the first term of the r.h.s. of \eqref{eq:corr1}. For the second term of the r.h.s. we proceed similarly and find 
\begin{align}
\int_{t_0}^t &  \; ds \; \vert \langle  e^{-i (s-t_0) d\Gamma ( \Delta)} \mathcal{W}_2 (t_0;0) \Omega,  \int dxdy \big[  K_s (x;y) a_x^*a_y^* + {\rm h.c.}\big]   \mathcal{W}_2 (s;0)  \Omega \rangle  \vert \notag \\
\leq& \int_{t_0}^t ds \; \|   e^{-i (s-t_0) d\Gamma ( \Delta)} \mathcal{W}_2 (t_0;0) \Omega - \mathcal{W}_2 (s;0)  \Omega \|  \; \| \widetilde{K}_{2,s} \|_{\rm HS} \; \| ( \mathcal{N}  +1 ) \mathcal{W}_{2}(s;0) \Omega \| \notag \\
\leq& C \int_{t_0}^t ds \; \|   e^{-i (s-t_0) d\Gamma ( \Delta)} \mathcal{W}_2 (t_0;0) \Omega - \mathcal{W}_2 (s;0)  \Omega \|  \frac{1}{(1+ \vert s \vert)^{3/2}} \label{eq:corr3}
\end{align}
where we concluded again by Lemma \ref{lemma:K} and Lemma \ref{lemma:bogo}. Summarizing, \eqref{eq:corr1}, \eqref{eq:corr2} and \eqref{eq:corr3} imply that 
\begin{align}
\|   & e^{-i (t-t_0)  d\Gamma ( \Delta)} \mathcal{W}_2 (t_0;0) \Omega - \mathcal{W}_2 (t;0)  \Omega \|^2 \notag \\
& \leq C \int_{t_0}^t ds \; \|   e^{-i (s-t_0) d\Gamma ( \Delta)} \mathcal{W}_2 (t_0;0) \Omega - \mathcal{W}_2 (s;0)  \Omega \|  \frac{1}{(1+ \vert s \vert)^{3/2}} 
\end{align}
and therefore by a Gronwall type argument that 
\begin{align}
\|   e^{-i (t-t_0) d\Gamma ( \Delta)} \mathcal{W}_2 (t_0;0) \Omega - \mathcal{W}_2 (t;0)  \Omega \|^2  \leq \frac{C}{(1+ \vert t_0 \vert)^{1/2}} \leq \frac{C}{N}
\end{align}
for all $t_0 \geq C  N^2$ which implies the desired estimate.  
\end{proof}

\bigskip

\noindent{\bf Acknowledgments.}  The research of P.T.N. and S.R. is partially funded by the European Research Council via the ERC Consolidator Grant RAMBAS (Project Nr. 10104424) and the Deutsche Forschungsgemeinschaft (TRR 352, Project Nr. 470903074). A.S. is partially supported by  NSF-DMS (No. 2205931).
Part of this work was carried out while the third author was visiting LMU, and he thanks the host for the invitation and support.

\section*{Declarations}

\subsection*{Conflict of interest} The authors declare that they have no conflict of interest.

\subsection*{Data availability} Data sharing not applicable to this article as no datasets were generated or analysed.


\begin{thebibliography}{}



\bibitem{Wieman} M. H. Anderson, J. R. Ensher, M. R. Matthews, C. E. Wieman, and E. A. Cornell. Observation of Bose-Einstein
condensation in a dilute atomic vapor. Science 269, pp. 198–201, 1995.

\bibitem{BGM} C. Bardos, F. Golse, N.J. Mauser, Weak coupling limit of the $n$-particle Schr\"odinger equation. Methods Appl. Anal. 7(2), 275-294, 2000. 

\bibitem{Bogo} Bogoliubov, N. On the theory of superfluidity. J. Phys. USSR XI, 23, 1947
.

\bibitem{Bose} S. Bose. Plancks Gesetz und Lichtquantenhypothese. Z. Phys. 26, pp. 178–181, 1924.

\bibitem{BPPS} L. Bossmann, S. Petrat, P. Pickl, and A. Soffer, Beyond Bogoliubov dynamics, Pure Appl. Anal. 3(4):677-726, 2021.

\bibitem{CW} T. Cazenave, F. Weissler, Rapidly decaying solutions of the nonlinear Schr\"odinger equation. Commun. Math. Phys. 147, 75-100, 1992. 


\bibitem{CL} L. Chen, J.O. Lee, Rate of convergence in nonlinear Hartree dynamics with factorized initial data. J. Math. Phys. 52(5), 052108, 25, 2011. 

\bibitem{CLS} L. Chen, J.O. Lee, B. Schlein, Rate of convergence towards Hartree dynamics. J. Stat. Phys. 144(4), 872-903, 2011. 

\bibitem{CLL} L. Chen, J.O. Lee, J. Lee, Rate of convergence toward Hartree dynamics with singular interaction potential. J. Math. Phys. 59(3), 031902, 2018. 

\bibitem{CO} Y. Cho, T. Ozawa, Small data scattering of Hartree type fractional Schr\"odinger equations in dimension 2 and 3. J. Korean Math. Soc. 55(2), 373, 
2018.

\bibitem{CLS} J. Chong, J. Lee, Z. Sun, Quantitative derivation of the two-component Gross-Pitaevskii equation in the hard-core limit with uniform-in-time convergence rate, Preprint, arXiv:2501.18787. 

\bibitem{Ketterle} K. B. Davis, M. O. Mewes, M. R. Andrews, N. J. van Druten, D. S. Durfee, D. M. Kurn, and W. Ketterle. Bose-
Einstein Condensation in a Gas of Sodium Atoms. Phys. Rev. Lett. 75, pp. 3969–3973, 1995.


\bibitem{DZ} P. Deift, X. Zhou, Long-time asymptotics for solutions of the NLS equation with initial data in a weighted Sobolev space. Commun. Pure Appl. Math. 56(8), 1029-1077, 2003. 

\bibitem{DL} C. Dietze and J. Lee, Uniform in Time Convergence to Bose-Einstein Condensation for a Weakly Interacting Bose Gas with an External Potential, Quantum Mathematics II. Proceedings of INdAM Quantum Meetings, pp. 267-311, 2023. 

\bibitem{DHR} T. Duyckaerts, J. Holmer, S. Roudenko, Scattering for the non-radial 3d cubic nonlinear Schr\"odinger equation. arXiv:0710.3630, 2007. 

\bibitem{Einstein} A. Einstein. Quantentheorie des einatomigen idealen Gases. Sitzungsberichte der Preußischen Akademie der Wis-
senschaften XXII, pp. 261–267, 1924, Zweite Abhandlung. Sitzungsber.
Preuss. Akad. Wiss., Phys.-math. Klasse (1925), 3-14.



\bibitem{ESY} L. Erd\H{o}s, M. Salmhofer,  H.T. Yau, {On the Quantum Boltzmann Equation}. J. Stat. Phys. 116, 367-380, 2004. 

\bibitem{ES} L. Erd\H{o}s, B. Schlein, Quantum dynamics with mean field interactions: a new approach. J. Stat. Phys. 134(5-6), 859-870, 2009. 


\bibitem{EY}  L. Erd\H{o}s,  H.T. Yau, Derivation of the nonlinear Schr\"odinger equation from a many body Coulomb system. Adv. Theor. Math. Phys. 5(6), 1169-1205, 2001. 

\bibitem{GO} J. Ginibre, T. Ozawa, Long range scattering for nonlinear Schr\"odinger and Hartree equations in space dimension $n\geq  2$. Commun. Math. Phys. 151(3), 619, 1993. 

\bibitem{GV3} J. Ginibre, G. Velo, The classical field limit of scattering theory for nonrelativistic many-boson systems. I. Commun. Math. Phys. 66(1), 37-76, 1979. 

\bibitem{GV4} J. Ginibre, G. Velo, The classical field limit of scattering theory for nonrelativistic many-boson systems. II. Commun. Math. Phys. 68(1), 45-68, 1979. 

\bibitem{GV1} J. Ginibre, G. Velo, Time decay of finite energy solutions of the non linear Klein-Gordon and Schr\"odinger equations. In: Annales de l'Institut Henri Poincare. Physique theorique, vol. 43, no. 4, pp. 399-442, 1985. 

\bibitem{GV2} J. Ginibre, G. Velo, Scattering theory in the energy space for a class of Hartree equations. arXiv:math/9809183, 1998. 

\bibitem{GM1} M. Grillakis and M. Machedon, {Pair excitations and the mean field approximation of interacting Bosons, I}, Commun. Math. Phys. 324, 601-636, (2013). 

\bibitem{GM2} M. Grillakis and M. Machedon, Pair excitations and the mean field approximation of interacting Bosons, II, Comm. PDE 42, 24-67, 2017.

\bibitem{GMM1} M. G. Grillakis, M. Machedon, D. Margetis, Second-order corrections to mean field evolution of weakly interacting bosons. I, Commun. Math. Phys. 294, 273-301, 2010. 

\bibitem{GMM2} M. G. Grillakis, M. Machedon, D. Margetis, Second-order corrections to mean field evolution of weakly interacting bosons. II, Adv. Math. 228, 1788-1815, 2011.

\bibitem{HT} N. Hayashi, Y. Tsutsumi, Scattering theory for Hartree type equations. In: Annales de l'Institut Henri Poincare, Physique theorique, 1987. 

\bibitem{H} K. Hepp, The classical limit for quantum mechanical correlation functions. Commun. Math. Phys. 35, 265-277, 1974. 

\bibitem{KP} A. Knowles, P. Pickl, Mean-field dynamics: singular potentials and rate of convergence. Commun. Math. Phys. 298(1), 101-138, 2010. 

\bibitem{L} J. Lee, On the time dependence of the rate of convergence towards Hartree dynamics for interacting bosons. J. Stat. Phys. 176(2), 358-381, 2019. 

\bibitem{LNS} M. Lewin, P.T. Nam, B. Schlein, Fluctuations around Hartree states in the mean-field regime, Amer. J. Math. 137, 1613-1650, 2015.

\bibitem{LNSS} M. Lewin, P. T. Nam, S. Serfaty, and J. P. Solovej, Bogoliubov spectrum of interacting Bose gases, Comm. Pure Appl. Math. 68, 413-471, 2015.


\bibitem{MPP} D. Mitrouskas, S. Petrat, P. Pickl, Bogoliubov Corrections and Trace Dynamics of interacting bosons: Norm Convergence for the Hartree Dynamics, Rev. Math. Phys. 31, 1950024 2019.

\bibitem{N} M. Napi\'{o}rkowski, Dynamics of interacting bosons: a compact review. arXiv:2101.04594, 2021. 

\bibitem{NM1} P.T. Nam and M. Napi\'{o}rkowski, {A note on the validity of Bogoliubov correction to mean-field dynamics}. J. Math. Pures Appl. 108, 662-688, 2017. 

\bibitem{NM2} P.T. Nam and M. Napi\'{o}rkowski, Norm approximation for many-body quantum dynamics and Bogoliubov theory, Advances in Quantum Mechanics: contemporary trends and open problems, Springer, 2017. 

\bibitem{O} T. Ozawa, Long range scattering for nonlinear Schr\"odinger equations in one space dimension. Commun. Math. Phys. 139(3), 479-493, 1991.

\bibitem{P} P. Pickl, A simple derivation of mean field limits for quantum systems. Lett. Math. Phys. 97(2), 151-164, 2011. 

\bibitem{QPS} C. Qu, L.P. Pitaevskii, S. Stringari, Expansion of harmonically trapped interacting particles and time dependence of the contact, Phys. Rev. A 94, 063635, 2016.

\bibitem{R} S. Rademacher, {Central limit theorem for Bose gases interacting through singular potentials}, Lett. Math. Phys., 110, 2143- 2174, 2020.


\bibitem{RS} I. Rodnianski, B. Schlein, Quantum fluctuations and rate of convergence towards mean field dynamics. Commun. Math. Phys. 291(1), 31-61, 2009. 

\bibitem{Ross} J.A. Ross, P. Deuar, D.K. Shin et al. On the survival of the quantum depletion of a condensate after release from a magnetic trap. Sci Rep 12, 13178, 2022. 

\bibitem{S} H. Spohn: Kinetic equations from Hamiltonian dynamics: Markovian limits. Rev. Mod. Phys. 52(3), 569-615, 1980.

\bibitem{Ketterle2} J.M. Vogels, K. Xu, C. Raman, J.R. Abo-Shaeer, W. Ketterle, Experimental observation of the Bogoliubov transformation for a Bose–Einstein condensed gas. Phys. Rev. Lett. 88, 2002.

\bibitem{Ketterle3} K. Xu, Y. Liu, D. E. Miller, J. K. Chin, W. Setiawan, W. Ketterle, Observation of strong quantum depletion in a gaseous Bose–Einstein condensate. Phys. Rev. Lett. 96, 180405, 2006. 



\end{thebibliography}
\end{document}